\documentclass[11pt]{article}
\usepackage[margin=1.25in]{geometry}
\usepackage{mathpazo}

\usepackage{caption}
\usepackage{subcaption}
\usepackage{xcolor}
\usepackage{xspace}
\usepackage{euscript}
\usepackage{amstext}
\usepackage{amsmath}
\usepackage{amssymb}
\usepackage{graphicx}
\usepackage{paralist}
\usepackage{theorem}
\usepackage[]{hyperref}
\usepackage[final]{microtype}
\usepackage{picins}

\usepackage{algorithm}
\usepackage{algpseudocode}
\usepackage{listings}
\usepackage{color}

\usepackage{fixltx2e}
\usepackage{xspace}
\usepackage{adjustbox}

\newcommand{\pnt}{p}

\renewcommand{\Re}{\mathbb{R}}
\newcommand{\Sphere}{\mathbb{S}}
\newcommand{\SphereD}{\mathbb{S}^{d-1}}
\newcommand{\SphereDP}{\mathbb{S}^{d-1}_+}
\newcommand{\PntSet}{\mathsf{P}}
\newcommand{\Probl}{RMS problem }
\newcommand{\SbSet}{\mathsf{Q}}

\newcommand{\HprSet}{\mathsf{H}}

\newcommand{\cardin}[1]{| #1 | }
\newcommand{\floor}[1]{\left \lfloor #1 \right \rfloor}

\newcommand{\norm}[1]{\| #1 \|}
\newcommand{\magn}[1]{\| #1 \|}

\newcommand{\dotp}[2]{\langle #1, #2 \rangle}
\newcommand{\wtv}{u}

\newcommand{\Score}{\omega}
\newcommand{\kPoint}{\varphi}
\newcommand{\kSet}{\Phi}
\newcommand{\dirSet}{\Omega}
\newcommand{\AllDir}{\mathbb{U}}
\newcommand{\hpr}{\mathsf{h}}
\newcommand{\ray}{w}
\newcommand{\cell}{C}
\newcommand{\cellSet}{\EuScript{C}}
\newcommand{\NP}{\mathrm{NP}}
\newcommand{\wid}{\mathop{\mathrm{width}}}
\newcommand{\Sets}{\EuScript{R}}
\newcommand{\Basis}{B}
\newcommand{\SetEps}[1]{R_{#1}}
\newcommand{\AllSets}{\EuScript{R}_{\wtv}}
\newcommand{\Net}{\mathsf{N}}
\newcommand{\AllSetsNet}{\EuScript{R}_{\Net}}

\newcommand{\errorProb}{$\min$-error}
\newcommand{\sizeProb}{$\min$-size}

\newcommand{\minError}[1]{\ell(#1)}
\newcommand{\minSize}[1]{s(#1)}
\newcommand{\minSz}{s_{\epsilon}}

\newcommand{\posPoints}{\mathbb{X}}
\newcommand{\allPoints}{\Re^d}

\newcommand{\GreedyKI}{NSLLX}
\newcommand{\GreedyKII}{CTVW}
\newcommand{\HS}{HS}
\newcommand{\CoreSet}{RRS}

\newcommand{\abs}[1]{\left | #1 \right |}
\newcommand{\polylog}{\mathop {\mathrm{polylog}}}

\newcommand{\argmin}{\arg\!\min}

\newcommand{\seclab}[1]{\label{section:#1}}
\newcommand{\secref}[1]{Section~\ref{section:#1}}

\newcommand{\lemlab}[1]{\label{lemma:#1}}
\newcommand{\lemref}[1]{Lemma~\ref{lemma:#1}}

\newcommand{\figlab}[1]{\label{fig:#1}}
\newcommand{\figref}[1]{Figure~\ref{fig:#1}}

\newcommand{\algolab}[1]{\label{algorithm:#1}}
\newcommand{\algoref}[1]{Algorithm~\ref{algorithm:#1}}

\newcommand{\Eqlab}[1]{\label{equation:#1}}
\newcommand{\Eqref}[1]{Equation~\ref{equation:#1}}

\newcommand{\tablelab}[1]{\label{table:#1}}

\newtheorem{theorem}{Theorem}[section]
\newtheorem{lemma}[theorem]{Lemma}
\newtheorem{corollary}[theorem]{Corollary}
\newenvironment{proof}{\trivlist\item[]\emph{Proof}:}%
                  {\unskip\nobreak\hskip 1em plus 1fil\nobreak%
                           \rule{2mm}{2mm}
                           \parfillskip=0pt%
                           \endtrivlist}

\newcommand{\remove}[1]{}

\def\mparagraph#1{\par\medskip\noindent\textbf{#1.}\quad}

\makeatletter
\long\def\@makecaption#1#2{
   \vskip 10pt
   \setbox\@tempboxa\hbox{{\footnotesize \textbf{#1.} #2}}
   \ifdim \wd\@tempboxa >\hsize         
       {\footnotesize \textbf{#1.} #2\par}
     \else                              
       \hbox to\hsize{\hfil\box\@tempboxa\hfil}
   \fi}
\makeatother

\begin{document}
\title{Efficient Algorithms for k-Regret Minimizing Sets\thanks{
Work by Agarwal and Sintos is supported by NSF under grants
  CCF-15-13816, CCF-15-46392, and IIS-14-08846, by ARO grant W911NF-15-1-0408, and by Grant 2012/229 from the U.S.-Israel
  Binational Science Foundation.
Work by Suri and Kumar is supported by NSF under grant CCF-15-25817.}}
\date{}

\remove{
\author{
\alignauthor
Pankaj K. Agarwal\\
       \affaddr{Duke University}\\
\alignauthor
Nirman Kumar\\
       \affaddr{University of Memphis}\\
\alignauthor
Stavros Sintos\\
       \affaddr{Duke University}\\
\and  
\alignauthor
Subhash Suri\\
       \affaddr{UC Santa Barbara}\\
}
}

\author{Pankaj K. Agarwal\thanks{%
Department of Computer Science, Duke University, Durham, NC
27708-0129, USA; {\tt pankaj@cs.duke.edu}.}
\and
Nirman Kumar\thanks{%
Department of Computer Science, University of Memphis, Memphis, TN
38152, USA; {\tt nkumar8@memphis.edu}.}
\and
Stavros Sintos\thanks{%
Department of Computer Science, Duke University, Durham, NC
27708-0129, USA; {\tt ssintos@cs.duke.edu}.}
\and
Subhash Suri\thanks{%
Department of Computer Science, University of California, Santa Barbara, CA
93106, USA; {\tt suri@cs.ucsb.edu}.}
}

\maketitle

\begin{abstract}
A regret minimizing set $\SbSet$ is a small size representation of a much larger database
$\PntSet$ so that user queries executed on $\SbSet$ return answers whose scores are not much
worse than those on the full dataset. In particular, a \emph{$k$-regret minimizing set}
has the property that the regret ratio between the score of the top-$1$ item in $\SbSet$
and the score of the top-$k$ item in $\PntSet$ is minimized, where the score of an item is
the inner product of the item's attributes with a user's weight (preference) vector.
The problem is challenging because we want to find a \emph{single} representative
set $\SbSet$ whose regret ratio is small with respect to \emph{all possible} user weight
vectors.

We show that $k$-regret minimization is $\NP$-Complete for all dimensions $d \geq 3$.
This settles an open problem from Chester et al.~[VLDB 2014], and resolves the
complexity status of the problem for all $d$: the problem is known to have
polynomial-time solution for $d \leq 2$. In addition, we propose two new
approximation schemes for regret minimization, both with provable guarantees,
one based on coresets and another based on hitting sets.
We also carry out extensive experimental evaluation, and show that our schemes
compute regret-minimizing sets comparable in size to the greedy algorithm
proposed in~[VLDB 14] but our schemes are significantly faster and scalable
to large data sets.
\end{abstract}

\section{Introduction}
\seclab{introduction}

Multi-criteria decision problems pose a unique challenge for databases systems:
how to present the space of possible answers to a user. In many instances, there
is no single best answer, and often a very large number of incomparable objects
satisfy the user's query. For instance, a database query for a car or a smart
phone can easily produce an overwhelming number of potential choices to present
to the user, with no obvious way to rank them. Top-$k$ and the skyline operators
are among the two main techniques used in databases to manage this kind of
complexity, but each has its own shortcoming.

The top-$k$ operator relies on the existence of a \emph{utility function} that is used
to rank the objects satisfying the user's query, and then selecting the top $k$ by score
according to this function. A commonly used utility function takes the inner product
of the object attributes with a \emph{weight vector}, also called the \emph{user's
preference}, thus forming a weighted linear combination of the different features.
However, formulating the utility function is complicated, as users often do not know
their preferences precisely, and, in fact, exploring the cost-benefit tradeoffs of
different features is often the goal of database search.

The second approach of skylines is based on the principle of \emph{pareto optimality}:
if an object $p$ is better than another object $q$ on all features, then $p$ is always
preferable to $q$ by any rational decision maker. This coordinate-wise dominance is
used to eliminate all objects that are dominated by some other object. The \emph{skyline}
is the set of objects not dominated by any other object, and has proved to be a powerful
tool in multi-criteria optimization. Unfortunately, while skylines are extremely
effective in reducing the number of objects in low dimensions, their utility drops
off quickly as the dimension (number of features) grows, especially when objects in the database
have anti-correlated features (attributes).
Indeed, theoretically all objects of the database can appear on the
skyline even in two dimensions.
Furthermore, the skyline does not necessarily preserve "top-$k$" objects as $k$ increases, in which case one uses
\emph{$k$-skybands} (\cite{gong2009efficient, liu2012efficient} -- the subset of points each of which is dominated by at most $k$ points.
The size of the skyband grows even more rapidly.

Regret minimization is a recent approach, proposed initially by Nanongkai et al.~\cite{nanongkai2010regret},
to address the shortcomings of both the top $k$ and skylines. The regret minimization
hybridizes top $k$ and skylines by computing a small representative subset $\SbSet$ of the
much larger database $\PntSet$ so that \emph{for any preference vector} the top ranked item
in $\SbSet$ is a good approximation of the top ranked item in $\PntSet$.
The hope is that the size of $\SbSet$ is much smaller than that of the skyline of $\PntSet$. In fact,
it is known that for a given regret ratio, there is always a regret minimizing set whose size depends only on the
regret ratio and the dimension, and not on the size of $\PntSet$. In contrast, as mentioned above, the
skyline size can be as large as $\cardin{\PntSet}$.

The goal is to find a
subset $\SbSet$ of small size whose approximation error is also small: posed in the form of
a decision question, is there a subset of $r$ objects so that every user's top-$1$
query can be answered within error at most $x \%$? In general, this is too stringent
a requirement and motivated Chester et al.~\cite{chester2014computing} to propose a more relaxed version of
the problem, called the $k$-regret minimization.\footnote{%
	We should point out that the term $k$-regret is used to denote different
	things by Nanongkai et al.~\cite{nanongkai2010regret} and Chester et al.~\cite{chester2014computing}.
	In the former, $k$-regret is the representative set of $k$ objects, whereas
	in the latter, $k$-regret is used to denote the regret ratio between the scores
	of top $1$ and top $k$. In our paper, we follow the convention of Chester et al.~\cite{chester2014computing}.}
In $k$-regret minimization, the quality of approximation is measured as the gap between
the score of the top $1$ item in $\SbSet$ and the top $k$ item in $\PntSet$ expressed as a ratio, so
that the value is always between $0$ and $1$.

In this paper, we make a number of contributions to the study of $k$-regret minimizing
sets. As a theoretical contribution, we prove that the $k$-regret minimization problem
is $\NP$-Complete for any dimension $d \geq 3$. This resolves an open problem of Chester
et al.~\cite{chester2014computing} who presented a polynomial-time algorithm for $d=2$ and showed $\NP$-hardness
for dimension $d \geq n$, leaving open the tantalizing question of whether the problem
was in class $\mathrm{P}$ for low dimensions --- the dimension being a fixed constant. Our result shows otherwise
and settles the complexity landscape of the problem for all dimensions.
On more practical side, we present simple and efficient algorithms
that are guaranteed to compute small regret minimizing sets and that are scalable to
large datasets even for larger values of $k$ and even when the size of skyline is large.

\mparagraph{Our Model}
An object is represented as a point $\pnt=(\pnt_1,\ldots, \pnt_d)$ in $\Re^d$ with non-negative
attributes, i.e., $\pnt_i\geq 0$ for every $i\leq d$. Let
$\posPoints=\{(\pnt_1,\ldots,\pnt_d)\in \Re^d \mid \pnt_i\geq 0 \enskip \forall i\}$ denote
the space of all objects, and let $\PntSet\subset \posPoints$ be a set of $n$ objects.
A user preference is also represented as a point $\wtv=(\wtv_1,\ldots, \wtv_d)\in \posPoints$,
i.e., all $\wtv_i\geq 0$. Given a preference $\wtv \in \Re^d$, we define the \emph{score} of an object
$\pnt$ to be $\Score(\wtv,\pnt) = \dotp{\wtv}{\pnt} = \sum_{i=1}^d \wtv_{i}\pnt_i$.

\begin{figure}[ht!]
\centering
\includegraphics[scale=0.25]{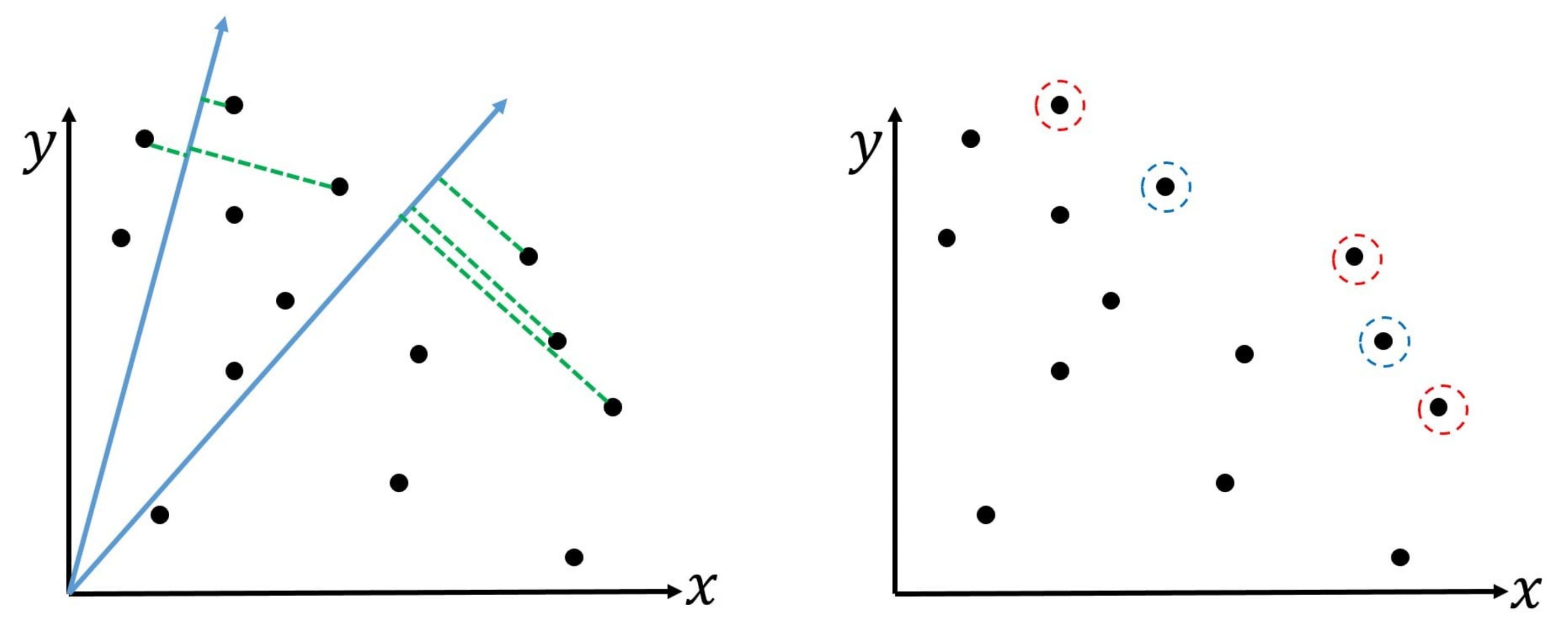}
\caption{Left: top $3$ points in two different preferences. Right: Set of points in the red circles is a $(1,0)$-regret set. Set of points in the blue circles is a $(3,0)$-regret set.}
\figlab{RegretExample}
\end{figure}


For a preference $\wtv \in \posPoints$ and an integer $k \geq 1$, let $\kPoint_k(\wtv,\PntSet)$
denote the point $\pnt$ in $\PntSet$ with the $k$-th largest score (i.e., there are less than $k$
points of $\PntSet$ with larger score than $\Score(\wtv,\pnt)$ and there are at least $k$ points
with score at least $\Score(\wtv,\pnt)$), and let $\Score_k(\wtv,\PntSet)$
denote its score. Set $\kSet_k(\wtv,\PntSet)=\{\kPoint_j(\wtv,\PntSet)\mid 1\leq j\leq k\}$
to be the set of $k$ top points with respect to preference $\wtv$.
\footnote{If there are multiple objects with score $\Score_j(\wtv,\PntSet)$,
then either we include all such points in $\kSet_k(\wtv,\PntSet)$ or break the tie
in a consistent manner.} For brevity, we set $\Score(\wtv,\PntSet)=\Score_1(\wtv,\PntSet)$.
If $\PntSet$ is obvious from the context, we drop $\PntSet$ from the list of the arguments,
i.e., we use $\Score_k(\wtv)$ to denote $\Score_k(\wtv,\PntSet)$ and so on.

For a subset $\SbSet \subseteq \PntSet$ and a preference $\wtv$, define the regret of
$\SbSet$  for preference $\wtv$ (w.r.t. $\PntSet$, denoted by $\ell_k(\wtv,\SbSet,\PntSet)$, as
\[
 \ell_k(\wtv,\SbSet,\PntSet)=\frac{\max\{0,\Score_k(\wtv,\PntSet)-\Score(\wtv,\SbSet)\}}{\Score_k(\wtv,\PntSet)}.
\]
That is, $\ell_k(\wtv,\SbSet,\PntSet)$ is the relative loss in the score of the $k$-th topmost object
if we replace $\PntSet$ with $\SbSet$. We refer to the maximum regret of $\SbSet$
\[
  \ell_k(\SbSet,\PntSet)=\max_{\wtv\in\posPoints} \ell_k(\wtv,\SbSet,\PntSet)
\]
as the \emph{regret ratio} of $\SbSet$ (w.r.t. $\PntSet$). If
$\ell_k(\SbSet)\leq \epsilon$, we refer to $\SbSet$ as a $(k,\epsilon)$-regret set (see \figref{RegretExample}).
By definition, a $(k,\epsilon)$-regret set is also a $(k',\epsilon)$-regret set for
any $k'\geq k$. In particular, a $(1,\epsilon)$-regret set is a $(k,\epsilon)$-regret
set for any $k\geq 1$. However, there may exist a $(k,\epsilon)$-regret set whose size
is much smaller than any $(k-1,\epsilon)$-regret set, so the notion of $(k,\epsilon)$-regret
set is useful for all $k$.

Notice that $\ell_k(\SbSet)$ is a monotonic decreasing function of its argument, i.e.,
if $\SbSet_1 \subseteq \SbSet_2$, then $\ell_k(\SbSet_1) \geq \ell_k(\SbSet_2)$.
Furthermore, for any $t>0$, $\Score(t\wtv,\pnt)=t\Score(\wtv,\pnt)$ but
$\kPoint_k(t\wtv,\PntSet)=\kPoint_k(\wtv, \PntSet)$, $\kSet_k(t\wtv, \PntSet) = \kSet_k(\wtv, \PntSet)$,
and $\ell_k(t\wtv,\SbSet,\PntSet)=\ell_k(\wtv,\SbSet,\PntSet)$ (scale invariance).

Our goal is to compute a small subset $\SbSet\subseteq \PntSet$ with small regret ratio,
which we refer to as the \emph{regret minimizing set} (RMS) problem. Since the regret ratio can be
decreased by increasing the size of the subset, there are two natural formulations of the \Probl.
\begin{enumerate}[(i)]
  \item \emph{\errorProb{}}: Given a set $\PntSet$ of objects and a positive integer $r$,
  compute a subset of $\PntSet$ of size $r$ that minimizes the regret ratio, i.e., return a subset
  \[
   \SbSet^* = \argmin_{\SbSet\subseteq\PntSet: \cardin{\SbSet}\leq r}\ell_k(\SbSet),
  \]
  and let $\minError{r}=\ell_k(\SbSet^*)$,
  \item \emph{\sizeProb{}}: Given a set $\PntSet$ of objects and a parameter $\epsilon>0$,
  compute a smallest size subset with regret ratio at most $\epsilon$, i.e., return
  \[
   \SbSet^{\#}=\argmin_{\SbSet\subseteq \PntSet: \ell_k(\SbSet)\leq \epsilon}\cardin{\SbSet},
  \]
  and set $\minSize{\epsilon}=\cardin{\SbSet}$.
\end{enumerate}

\mparagraph{Our results}
We present the following results in this paper:\\

(I)  We show that the \Probl is NP-Complete even for $d=3$ and $k>1$.
The previous hardness proof \cite{chester2014computing} requires the dimension $d$ to be as
large as $n$, and it was an open question whether the problem was NP-complete in low dimensions.
Since a polynomial-time algorithm exists for both formulations of the regret minimizing set problem
in $2D$, our result settles the problem for $k>1$.
Proving hardness in small dimensions, $d=3$, requires a different proof technique.
In  fact, it is not trivial to check whether $\ell_k(\SbSet)\leq \epsilon$ for given $\epsilon>0$,
i.e., it is not obvious that the \Probl is in NP. Using a few results from discrete geometry,
we present an efficient algorithm for computing $\ell_k(\SbSet)$.\\

(II) We show that for any $\PntSet\subset \posPoints$ and for any $\epsilon>0$
there exists a $(1,\epsilon)$-regret set, and thus a $(k,\epsilon)$-regret set for any $k\geq 1$,
of $\PntSet$ whose size is independent of the size of $\PntSet$.

By establishing a connection between $(k,\epsilon)$-regret sets and the so-called core sets \cite{agarwal2005geometric},
we show that for any $\PntSet\subset \posPoints$ and for any $\epsilon>0$,
a $(1,\epsilon)$-regret set of size $O(\frac{1}{\epsilon^{(d-1)/2)}})$
can be computed in time $O(n+\frac{1}{\epsilon^{d-1}})$. Notice that for the \errorProb{} problem
Nanongkai et al. \cite{nanongkai2010regret} give an algorithm that returns a set $\SbSet$
such that $\minError{r}\leq \frac{d-1}{(r-d+1)^{\frac{1}{d-1}}+d-1}$. Solving for $r$, we get
for a fixed error $\epsilon$ a $(1,\epsilon)$-regret set of size
$O(\frac{1}{\epsilon^{(d-1)}})$. Our result improves this bound significantly
and it is optimal in the worst case. Furthermore, we can maintain our
$(1,\epsilon)$-regret set under insertion/deletion of points in $O(\frac{\polylog(n)}{\epsilon^{d-1}})$ time per update.
The efficient maintenance of a regret set is important in various applications and it has not been considered before.\\

(III) For a given $\PntSet$ and $\epsilon>0$, there may exist a $(k,\epsilon)$-regret set of $\PntSet$
of size much smaller than $1/\epsilon^{\frac{d-1}{2}}$.
We complement our NP-Completeness result by presenting approximation algorithms for the \Probl.
Given $\PntSet\subset \posPoints$ of size $n$ and $\epsilon>0$, we can compute a $(k,2\epsilon)$-regret set of $\PntSet$
\footnote{The approximation ratio $2$ is not important. We can actually compute a $(k,t\epsilon)$-regret set for an arbitrary
small constant $t>1$.} of size $O(\minSize{\epsilon}\log(\minSize{\epsilon}))$ in time $O(\frac{n}{\epsilon^{d-1}}\log(n)\log(\frac{1}{\epsilon}))$.
Roughly speaking, we formulate the regret-minimizing set problem as a classical hitting-set problem and
use a greedy algorithm to compute a small size hitting set.

By plugging the above algorithm into a binary search,
we also obtain an algorithm for the \errorProb{} version of the problem: given a parameter $r$,
we compute a set $\SbSet\subseteq \PntSet$ of size $O(r\log r)$ such that $\ell(cr\log r)\leq \ell_k(\SbSet)\leq 2\ell(r)$
for a sufficiently large constant $c$. The algorithm runs in $O(\frac{n}{(\ell_k(\SbSet))^{d-1}}\log(n)\log(\frac{1}{\ell_k(\SbSet)}))$ time.
If $\ell_k(\SbSet)$ is very small the algorithm runs in $O(n^d)$ time. The expected running time of this algorithm
is much smaller if the objects are uniformly distributed or drawn from some other nice distribution.\\

(IV) We present experimental results to evaluate the efficacy and
the efficiency of our algorithms on both synthetic and real data sets.
We compare our algorithms with the state of the art greedy algorithm for the
$k$-regret minimization problem presented in \cite{chester2014computing}.
Our hitting-set based algorithm is significantly faster than the previous known algorithms
and the maximum regret ratios of the returned sets are very close, if not better,
than the maximum regret ratios of the greedy algorithm.
The core set algorithm is significantly faster than hitting set and greedy algorithms.
Although the (maximum) regret ratio of the set returned by the core-set based
algorithm is worse than those of other algorithms, the regret in $90\%-95\%$ directions
is roughly the same as that of the other two algorithms.




\section{3D RMS is NP-Complete}
\seclab{NPC3D}
In this section we show that the $k$-\Probl is $\NP$-Complete
for $d\geq 3$ and $k\geq 2$. More precisely, given a set $\PntSet\subset \posPoints$
in $\Re^3$, a parameter $\epsilon>0$, and an integer $r$, the problem of determining
whether there is a $(k,\epsilon)$-regret set of $\PntSet$ of size at most $r$ is
$\NP$-Complete. We first show its membership in $\NP$. We then show
$\NP$-hardness for $k=2$ and later show how to extend the argument to higher values of $k$.

\subsection{RMS problem is in NP}
Given a subset $\SbSet\subseteq \PntSet$ of objects, we describe a polynomial-time
algorithm for computing the regret ratio of $\SbSet$. For simplicity, we describe the
algorithm for $d=3$ but it extends to $d>3$.

\begin{figure}[ht!]
\centering
\includegraphics[scale=0.25]{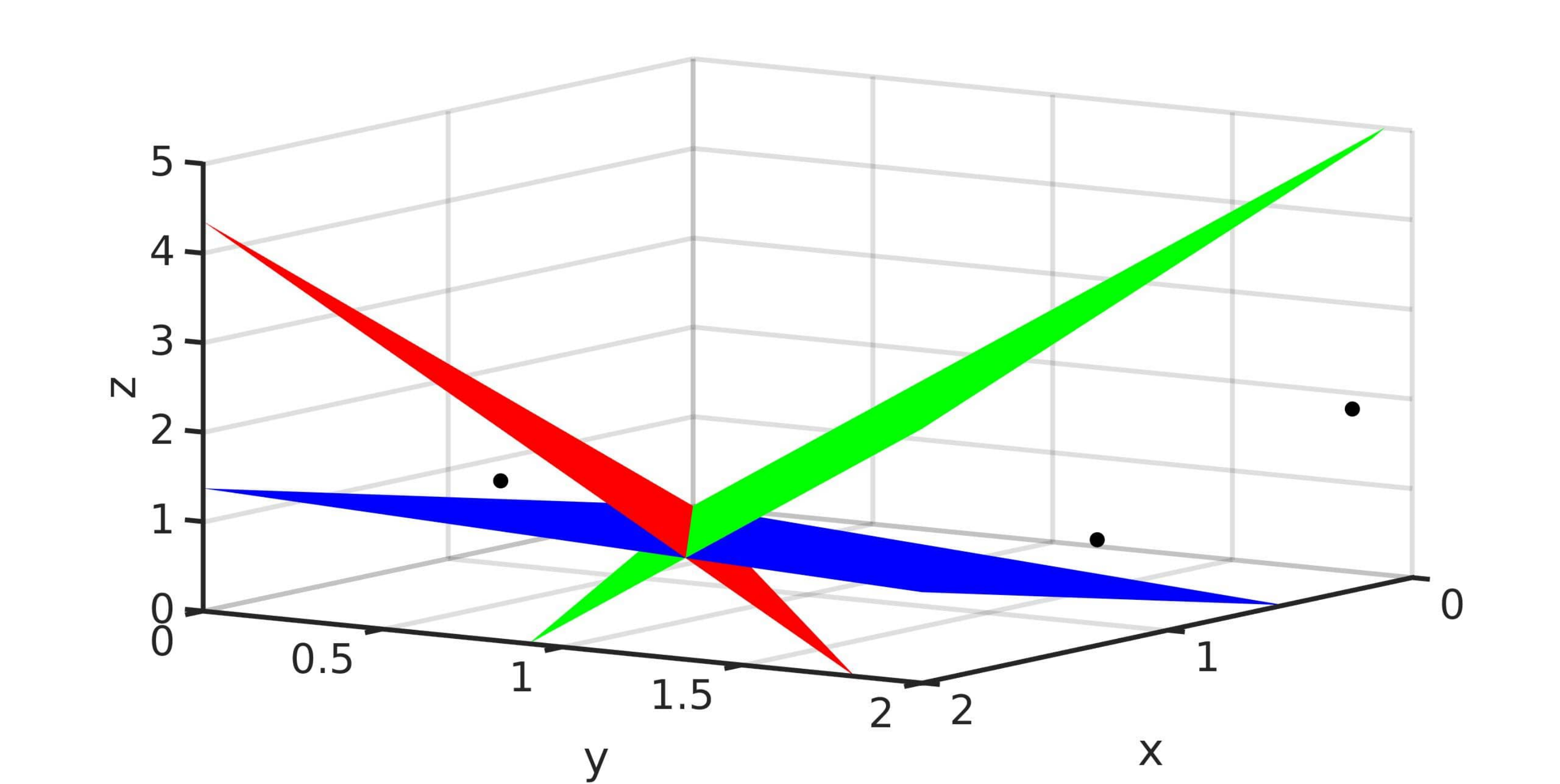}
\caption{$\HprSet$ for a set of $3$ points in $\Re^3$. \label{3dHprs}}
\end{figure}

Let $\dirSet=\{p-q\mid p, q \in\PntSet, p \neq q \}$ be the set of vectors in directions
passing through a pair of points of $\PntSet$. For a vector $\ray\in\dirSet$, let
$\hpr_{\ray}: \dotp{x}{\ray}=0$, be the plane normal to $\ray$ passing through the origin.
By construction, for $\ray=p-q$ the score of $p$ is higher than that of $q$ for all preferences in one
of the open halfspaces bounded by $\hpr_{\ray}$ (namely, $\dotp{x}{\ray}>0$), lower in the other halfspace, and
equal for all preferences in $\hpr_{\ray}$.
Set $\HprSet = \{\hpr_{\ray}\mid \ray\in \dirSet\}\cup\{x_i=0\mid 1\leq i\leq 3\}$, i.e.,
$\HprSet$ includes all the planes $\hpr_{\ray}$ along with the coordinate planes.
$\HprSet$ induces a decomposition $A(\HprSet)$ of $\Re^3$ into cells of various dimensions, where
each cell is a maximal connected region of points lying in the same subset of hyperplanes
of $\HprSet$ (see Figure~\ref{3dHprs}). It is well known that
\begin{enumerate}[(i)]
\item each cell of $A(\HprSet)$ is a polyhedral cone with the origin as
its apex (i.e., each cell is the convex hull of a finite set of rays, each emanating from the origin),and
\item the only $0$-dimensional cell of $A(\HprSet)$ is the origin itself, and
the $1$-dimensional cells are rays emanating from the origin.
Let $\cellSet\subseteq A(\HprSet)$ be the set of cells that lie in $\posPoints$, the positive orthant.
\end{enumerate}
For each cell $\cell\in \cellSet$, let $\ell(\cell,\SbSet) = \max_{\wtv\in \cell}\ell_k(\wtv,\SbSet)$ the regret
ratio of $\SbSet$ within $\cell$. Then $\ell_k(\SbSet) = \max_{\cell\in \cellSet}\ell_k(\cell,\SbSet)$. The following
lemma is useful in computing $\ell(\cell,\SbSet)$.
\begin{lemma}
\lemlab{LemmaA}
  For each cell $\cell\in A(\HprSet)$ and for any $i\leq n$, $\kPoint_i(\wtv,\PntSet)$ (and thus $\kPoint_i(\wtv,\SbSet)$)
  is the same for all $\wtv\in \cell$.
\end{lemma}
\begin{proof}
Suppose on the contrary, there are two points
 $\wtv_1,\wtv_2\in \cell$, and $j\geq 0$ such that
 $\kPoint_j(\wtv_1,\SbSet)\neq \kPoint_j(\wtv_2,\SbSet)$. Hence, there are
 two points $\pnt_1$, $\pnt_2\in \SbSet$ such that $\dotp{\wtv_1}{\pnt_1} \geq \dotp{\wtv_1}{\pnt_2}$
 and $\dotp{\wtv_2}{\pnt_1} \leq \dotp{\wtv_2}{\pnt_2}$, and at least one of the inequalities is strict.
 Let $\hpr_{\ray}\in \HprSet$ be the plane
 that is normal to $\pnt_1-\pnt_2$ and passes through the origin. It
 divides $\Re^3$ into two halfspaces. Reference vectors $\wtv_1$, $\wtv_2$ lie in the opposite halfspaces of $\hpr_{\ray}$,
 and at least one of the $\wtv_1, \wtv_2$ lies in the open halfspace.
 However, this is a contradiction because $\wtv_1,\wtv_2$ lie in the same cell of $A(\HprSet)$
 and thus lie on the same side of each plane in $\HprSet$.
\end{proof}
Fix a cell $\cell$.
Let $\pnt_i=\kPoint_1(\wtv,\SbSet)$ and $\pnt_j=\kPoint_k(\wtv,\PntSet)$ for any
$\wtv\in \cell$ (from \lemref{LemmaA} we have that the ordering inside a cell is the same).
Furthermore, let $\hpr_j$ be the plane $\dotp{x}{\pnt_j}=1$ and let $\cell^{\downarrow}=\hpr_j\cap \cell$.
$\cell^{\downarrow}$ is a $2$D polygon and each ray $\rho$ in $\cell$ intersects $\cell^{\downarrow}$
at exactly one point $\rho^{\downarrow}$. Since $\ell(\wtv,\SbSet)$ is the same for all points on $\rho$,
$\ell_k(\cell,\SbSet)=\ell_k(\cell^{\downarrow}, \SbSet)$. Furthermore, by \lemref{LemmaA},
$\ell_k(\cell^{\downarrow}, \SbSet)$ is either $0$ for all $\wtv\in \cell^{\downarrow}$ or
\begin{align*}
\ell_k(\cell^{\downarrow},\SbSet)&=\max_{\wtv\in\cell^{\downarrow}}\frac{\Score(\wtv,\pnt_j)-\Score(\wtv,\pnt_i)}{\Score(\wtv,\pnt_j)}
=\max_{\wtv\in\cell^{\downarrow}} 1-\Score(\wtv,\pnt_i)\\
& =1-\min_{\wtv\in\cell^{\downarrow}}\dotp{\wtv}{\pnt_i}.
\end{align*}
Since $\cell^{\downarrow}$ is convex and $\dotp{\wtv}{\pnt_i}$ is a linear function, it is a minimum within
$\cell^{\downarrow}$ at a vertex of $\cell^{\downarrow}$, so we compute $\dotp{\wtv}{\pnt_i}$ for each vertex
$\wtv\in \cell^{\downarrow}$ and choose the one with the minimum value. Repeating this step for all cells of $\cellSet$
we compute $\ell_k(\SbSet)$.

By a well known result in discrete geometry \cite{agarwal2000arrangements}, the total number of vertices in
$\cell^{\downarrow}$ over all cells $\cell\in\cellSet$ is $O(\cardin{\HprSet}^2)=O(n^4)$. Furthermore,
if $b$ bits are used to represent the coordinates of each point in $\PntSet$, each vertex of $\cell^{\downarrow}$
requires $O(b)$ bits. Finally, the algorithm extends to higher dimensions in a straightforward manner.
The total running time in $\Re^d$ is $O(n^{2d-1})$. We thus conclude the following.
\begin{lemma}
 \lemlab{polyt}%
 Given a set $\PntSet$ of $n$ points in $\Re^d$ and a subset $\SbSet\subseteq \PntSet$,
 $\ell_k(\SbSet)$ can be computed in $O(n^{2d-1})$ time.
\end{lemma}
An immediate corollary of the above lemma is the following:
\begin{corollary}
The \Probl is in $\NP$.
\end{corollary}

\subsection{NP-Hardness Reduction}
We first show the hardness for $k = 2$.
Recall that a preference vector has only non-negative coordinates.
For simplicity, however, we first consider
all points in $\Re^3$ as preference vectors and define
$\ell_k(\SbSet)=\max_{\wtv\in \Re^3} \ell_k(\wtv,\SbSet)$,
and later we describe how to restrict the preference vectors to $\posPoints$.

Recall that the \Probl for $\epsilon=0$ and $k=2$ asks:
Is there a subset $\SbSet \subseteq \PntSet$ of size $r$ such that
in every direction $\wtv$, the point in $\SbSet$ with the highest score along
$\wtv$, i.e., $\kPoint_1(\wtv,S)$, has score at least as much as that of the second best in $\PntSet$
along $\wtv$, i.e., of the point $\kPoint_2(\wtv,\PntSet)$?

Let $\Pi$ be a strictly convex polytope in $\Re^3$. The \emph{$1$-skeleton} of $\Pi$ is the graph
formed by the vertices and edges of $\Pi$. Given $\Pi$ and an integer $r>0$, the
\emph{convex-polytope vertex-cover} (CPVC) asks whether the $1$-skeleton of $\Pi$ has a vertex cover
of size at most $r$, i.e., whether there is a subset $C$ of vertices of $\Pi$ of size $r$ such that
every edge is incident on at least one vertex of $C$. The CPVC problem is $\NP$-Complete, as shown by Das and Goodrich
\cite{dg-copcp-97}.
Given $\Pi$ with $V$ as the set of its vertices, we construct an instance of the \Probl
for $k=2$, as follows. First we translate $\Pi$ so that the origin lies inside $\Pi$. Next we set $\PntSet=V$.
The next lemma proves the $\NP$-hardness of the \Probl for $k=2$ and $\epsilon=0$.

\begin{lemma}
A subset $\SbSet\subseteq V$ is a vertex cover of $\Pi$ if and only if $\SbSet$ is a
$(2,0)$-regret set for $\PntSet$.
\end{lemma}
\begin{proof}
If $\SbSet$ is a vertex cover of $\Pi$, we show that $\SbSet$ is also a $(2,0)$-regret set.
Take a vector $\wtv\in\Re^3$ and assume that $q=\kPoint_1(\wtv,\PntSet)$
(if there is more than one point with rank one, we can let $q$ be any one of them).
If $q\in \SbSet$ then obviously $\Score_1(\wtv,\SbSet)=\Score_1(\wtv,\PntSet)\geq \Score_2(\wtv,\PntSet)$.
Now, assume that $q\notin \SbSet$.
Let $(q,q_1),\ldots, (q,q_g)$ be the edges in $\Pi$ incident on $q$. Set $N_q=\{q_i\mid 1\leq i\leq g\}$.
Since $\SbSet$ is a vertex cover of $\Pi$ and $q\notin \SbSet$, $N_q\subseteq \SbSet$.
We claim that $\kPoint_2(\wtv,\PntSet)\in N_q$, which implies that $\Score(\wtv,\SbSet)\geq \Score_2(\wtv,\PntSet)$.
Hence, $\SbSet$ is a $(2,0)$-regret set.

Indeed, since $\Pi$ is convex, and $q$ is maximal along direction $\wtv$,
the plane $\hpr$ on $q$ vertical to $\wtv$ is a supporting hyperplane for $\Pi$. A plane $\hpr'$ parallel to $\hpr$ is translated toward the origin starting with its initial position at $\hpr$. There are two cases. In the first case, where $q$ and $\kPoint_2(\wtv,\PntSet)$ have the same score, they belong to the same face of $\Pi$ that must be contained in $\hpr$ itself --- in this case $\hpr$ also contains a point from $N_q$, since every face containing $q$ and points other than $q$ must contain a $1$ dimensional face as well, and therefore a point in $N_q$. In the second case, as $\hpr'$ is translated, it must first hit one of the
neighbors of $q$, by convexity. As a result, in any case, there will be a point in $N_q$ that
gives the rank-two point on $\wtv$.

Next, if $\SbSet$ is a $(2,0)$-regret set, we show that $\SbSet$ is a vertex cover of $\Pi$.
Suppose to the contrary $\SbSet$ is not a vertex cover of $\Pi$, i.e., there is an edge $(q_1,q_2)$
in $\Pi$ but $q_1, q_2\notin \SbSet$.
Since $\Pi$ is a strictly convex polytope, there is a plane $\hpr$ tangent to $\Pi$ at the edge $(q_1,q_2)$
that does not contain any other vertex of $\Pi$.
If we take the direction $\wtv$ normal to $\hpr$ then $\kSet_2(\wtv,\PntSet)=\{q_1, q_2\}$.
If $q_1, q_2\notin \SbSet$ then $\Score_1(\wtv,\SbSet)<\Score_2(\wtv,\PntSet)$, which contradicts the
assumption that $\SbSet$ is a $(2,0)$-regret set of $\PntSet$.
\end{proof}

\mparagraph{Restricting to $\posPoints$}
In order to show that the \Probl is $\NP$-hard
even when preferences are restricted to $\posPoints$, polytope
$\Pi$ needs to have two additional properties:
\begin{enumerate}[(i)]
\item All vertices of $\Pi$ must lie in the first orthant.
\item For any edge $(v_1, v_2)$ of $\Pi$, where $v_1, v_2$ are vertices of $P$, there is
a direction $\wtv\in \posPoints$ such that $v_1, v_2$ are the top vertices in direction $\wtv$.
\end{enumerate}
It is easy to satisfy property (i) because the translation of the vertices of a polytope
does not
change the rank of the points in any direction. On the other hand, property (ii)
is not guaranteed by the construction in \cite{dg-copcp-97}.

We show that there is an affine transformation of $\Pi$ that can be computed and applied in
polynomial time, to get a polytope $\Pi'$ with the same combinatorial structure as $\Pi$, but
that also satisfies properties (i), and (ii). The fact that the polytope has the same
combinatorial structure implies that the underlying graph is the same, and therefore a vertex
cover will also be a $(2,0)$-regret set of $\Pi$.
The details of the transformation can be found in Appendix~\ref{Ap2}.
The first part of the $\NP$-hardness proof is the same with the case of all directions in $\Re^3$,
if $\SbSet$ is a vertex cover of $\Pi'$ then it is also a $(2,0)$-regret set.
Using property (ii) of $\Pi'$, it is straightforward to show the other direction, as well.

\mparagraph{Choosing $\epsilon > 0$}
While the above suffices to prove the hardness of the \Probl for $\epsilon=0$,
it is possible that when $\epsilon > 0$ the problem is strictly easier.
However, we show the stronger result that the \Probl
is $\NP$-hard even when $\epsilon$ is required to be strictly positive.
In order to get the $\NP$-hardness of the \Probl for $\epsilon>0$ and $k=2$, we
find a small enough strictly positive  value of $\epsilon$ with bounded bit complexity such that
any $(2,\epsilon)$-regret set is also a $(2,0)$-regret set, and vice versa.
For each cell $\cell\in \cellSet$, we take a direction $\wtv_{\cell}\in \cell$ and let
$\lambda_{\cell}=1-\Score_3(\wtv_{\cell},\PntSet)/\Score_2(\wtv_{\cell},\PntSet)>0$. By defining
$\epsilon=\frac{1}{2}\min_{\cell}\lambda_c$ we can conclude the result.

\mparagraph{Larger values of $k$}
By making $k-1$ copies of each point in the above construction it is straightforward
to show that the \Probl is $\NP$-complete for any $k\geq 2$ and $d\geq 3$.

\begin{theorem}
 The \Probl is $\NP$-complete for $d\geq 3$ and for $k\geq 2$.
\end{theorem}

\section{Coreset-based Approximation} \seclab{coreSet}

In this section, we present an approximation scheme for the \Probl using coresets.
The general idea of a coreset is to approximately preserve some desired characteristics of
the full data set using only a tiny subset~\cite{agarwal2005geometric}.
The particular geometric characteristic most relevant to our problem is the \emph{extent} of
the input data in any direction, which can be formalized as follows. Given a set of points
$\PntSet$ and a direction $\wtv \in \allPoints$, the \emph{directional width} of $\PntSet$
along $\wtv$, denoted $\wid(\wtv,\PntSet)$, is the distance between the two supporting
hyperplanes of $\allPoints$, one in direction $\wtv$ and the other in direction $- \wtv$.
The connection between $k$-regret and the directional width comes from the fact that the
supporting hyperplane in a direction $\wtv$ is defined by the extreme point in that direction,
and its distance from the origin is simply its score. Therefore, we have the equality:
	$$\wid(\wtv,\PntSet) \:=\: \Score(\wtv,\PntSet)+\Score(-\wtv,\PntSet).$$

We use coresets that approximate directional width to approximate $k$-regret sets.
In particular, a subset $\SbSet\subseteq \PntSet$ is called an \emph{$\epsilon$-kernel} coreset if
$\wid(\wtv,\SbSet)\geq (1-\epsilon)\wid(\wtv,\PntSet)$, for all directions $\wtv\in \allPoints$.

\begin{lemma}
\lemlab{LemCoreSet}
If $\SbSet\subseteq \PntSet$ is an $\epsilon$-kernel coreset of $\PntSet$ then $\SbSet$
is also $(1,\epsilon)$-regret set of $\PntSet$.
\end{lemma}
\begin{proof}
If $\SbSet$ is  an $\epsilon$-kernel coreset of $\PntSet$ then
$\wid(\wtv,\PntSet)-\wid(\wtv,\SbSet)\leq \epsilon\wid(\wtv,\PntSet)\leq \epsilon\Score(\wtv,\PntSet)$.
The last inequality follows because $\Score(-\wtv,\PntSet)\leq 0$. Furthermore
$\Score(-\wtv,\SbSet)\leq \Score(-\wtv,\PntSet)$. We thus have
\begin{align*}
\Score(\wtv,\PntSet)-\Score(\wtv,\SbSet) &= \Score(\wtv,\PntSet)+\Score(-\wtv,\PntSet)-\Score(\wtv,\SbSet)-\Score(-\wtv,\PntSet)\\
& \leq \wid(\wtv,\PntSet)-\Score(\wtv,\SbSet)-\Score(-\wtv,\SbSet)\\
& = \wid(\wtv,\PntSet)-\wid(\wtv,\SbSet)\\
& \leq \epsilon\Score(\wtv,\PntSet).
\end{align*}
Hence, $\Score(\wtv,\SbSet)\geq (1-\epsilon)\Score(\wtv,\PntSet)$.
\end{proof}
We use the results of \cite{agarwal2004approximating, chan2004faster}
that compute small $\epsilon$-kernel coresets efficiently, as well as allows dynamic updates, and
prove the following result.

\begin{theorem}
\label{ThmCoreRegret}
Given a set $\PntSet$ of $n$ points in $\Re^d$, $\epsilon>0$ and an integer $k>0$, we can compute
in time $O(n+\frac{1}{\epsilon^{d-1}})$
a subset $\SbSet\subseteq \PntSet$ of size $O(\frac{1}{\epsilon^{(d-1)/2}})$ whose $k$-regret ratio
is at most $\epsilon$, i.e. $\ell_k (\SbSet, \PntSet)\leq \epsilon$.
\end{theorem}
\begin{proof}
We choose $\SbSet$ as an $\epsilon$-kernel coreset of $\PntSet$. By \lemref{LemCoreSet},
$\SbSet$ is a $(1,\epsilon)$-regret set of $\PntSet$ and thus also a $(k,\epsilon)$-regret set
of $\PntSet$ for any $k\geq 1$. Chan \cite{chan2004faster} has described an algorithm for computing
$\epsilon$-kernel of size $O(\frac{1}{\epsilon^{(d-1)/2}})$ in time $O(n+\frac{1}{\epsilon^{d-1}})$.
Hence, the theorem follows.
\end{proof}

We conclude this section by making two remarks:

The size of $\SbSet$ in the preceding theorem is asymptotically optimal: there exist point sets for
which no smaller subset satisfies this property. The size optimality follows from the construction
described below for $2$D---generalization to higher dimensions is straightforward.
Fix a positive integer $k\leq n\sqrt{\epsilon}$. Consider a set of $n/k$ points $\PntSet\subset \Sphere_+^{1}$
uniformly distributed on the unit circle, on the arc in the first quadrant.
Make $k$ copies of these points to get a set of $n$ points.
This construction ensures that the top-$k$ points in any direction $\wtv$ have exactly the same score;
that is, $\Score_1(\wtv,\PntSet)=\Score_k(\wtv,\PntSet)$ for all $\wtv\in \Re^2$. It is know that
when points are uniformly distributed on a circle, an $\epsilon$-kernel coreset gives the optimum
(asymptotically) coreset with size $\Theta(\frac{1}{\sqrt{\epsilon}})$. Since $\kPoint_1(\wtv,\PntSet')$
and $\kPoint_k(\wtv,\PntSet')$ lie at the same position for any $\wtv\in \Re^2$, we also have that
the optimal $(k,\epsilon)$-regret set has size $\Theta(\frac{1}{\sqrt{\epsilon}})$.

The set $\SbSet$ can also be maintained under insertion/deletion of points in $\PntSet$ in time
 $O(\frac{\log^{d}n}{\epsilon^{d-1}})$ per update.
The dynamic update performance follows from the construction in~\cite{agarwal2004approximating}.

\section{Regret Approximation using Hitting Sets} 	\seclab{approxAlg}
Theorem \ref{ThmCoreRegret} shows that a $(k,\epsilon)$-regret set of size
$O(\frac{1}{\epsilon^{(d-1)/2}})$ can be computed quickly. However, given $\PntSet$
and $\epsilon>0$, there may be a $(k,\epsilon)$-regret set of much smaller size.
In this section, we describe an algorithm that computes a $(k,\epsilon)$-regret set of size close
to $\minSz:=\minSize{\epsilon}$, the minimum size of a $(k,\epsilon)$-regret set,
by formulating the \Probl as a hitting-set problem.

A range space (or set system) $\Sigma = (\mathsf{X},\Sets)$ consists of a set $\mathsf{X}$
of objects and a family $\Sets$ of subsets of $\mathsf{X}$. A subset $H\subseteq \mathsf{X}$
is a hitting set of $\Sigma$ if $H\cap R\neq \emptyset$ for all $R\in \Sets$. The hitting
set problem asks to compute a hitting set of the minimum size. The hitting set problem is a
classical $\NP$-Complete problem, and a well-known greedy $O(\log n)$-approximation algorithm
is known.

We construct a set system $\Sigma=(\PntSet, \Sets)$ such that a subset $\SbSet\subseteq \PntSet$
is a $(k,\epsilon)$-regret set if and only if $\SbSet$ is a hitting set of $\Sigma$.
We then use the greedy algorithm to compute a small-size hitting set of $\Sigma$.
A weakness of this approach is that the size of $\Sets$ could be very large and the greedy
algorithm requires $\Sets$ to be constructed explicitly. Consequently, the approach is
expensive even for moderate inputs say $d\sim 5$.

Inspired by the above idea, we propose a bicriteria approximation algorithm: given $\PntSet$
and $\epsilon>0$, we compute a subset $\SbSet\subseteq \PntSet$ of size $O(\minSz\log\minSz)$
that is a $(k,2\epsilon)$-regret set of $\PntSet$; the constant $2$ is not important, it can be
made arbitrarily small at the cost of increasing the running time.
By allowing approximations to both the error and size concurrently, we can construct a much smaller range space
and compute a hitting set of this range space.

The description of the algorithm is simpler if we assume the input to be well conditioned.
We therefore transform the input set, without affecting an RMS, so that the score of the topmost point
does not vary too much with the choice of preference vectors, i.e., the ratio
$\frac{\max_{\wtv\in \posPoints}\Score(\wtv,\PntSet)}{\min_{\wtv\in \posPoints}\Score(\wtv,\PntSet)}$ is
bounded by a constant that depends on $d$.

We transform $\PntSet$ into another set $\PntSet'$, so that (i) for any
$\wtv\in \posPoints$, $\kPoint_1(\wtv,\PntSet')$ does not lie close to the origin and
(ii) for any $(k,\epsilon)$-regret set $\SbSet\subseteq \PntSet$, the subset
$\SbSet'\subseteq \PntSet'$ is a $(k,\epsilon)$-regret set in $\PntSet'$, and vice versa.
Nanongkai et al.~\cite{nanongkai2010regret} showed that a non-uniform scaling of
$\PntSet$ satisfies (ii). In the following lemma, we show a stronger result.
\begin{lemma}
\lemlab{Transfs}
 Let $\PntSet$ be a set of $n$ points in $\Re^d$, and let $M$ be a full rank
 $d\times d$ matrix. A subset $\SbSet\subseteq\PntSet$ is a
 $(k,\epsilon)$-regret set of $\PntSet$ if and only if $\SbSet'=M\SbSet$ is a
 $(k,\epsilon)$-regret set of $\PntSet'=M\PntSet$.
\end{lemma}
\begin{proof}
First, observe that $\dotp{\wtv}{M\pnt} = \wtv^TM\pnt = (M^T\wtv)^T\pnt =
\dotp{M^T\wtv}{\pnt}$, and so $\Score_k(\wtv,M\PntSet)=\Score_k(M^T\wtv,\PntSet)$.
We define a mapping $F:\posPoints \to \posPoints$ and its inverse $F^{-1}:\posPoints \to \posPoints$
as $F(\wtv)=(M^{-1})^T\wtv$ and $F^{-1}(\wtv)=M^T\wtv$.
Our proof now follows easily from these mappings.

 If $\SbSet$ is a $(k,\epsilon)$-regret set for $\PntSet$, then for any $\wtv\in \SphereDP$ we have
 $\Score_1(\wtv,M\SbSet)=\Score_1(M^T\wtv,\SbSet)=\Score_1(F^{-1}(\wtv),\SbSet)
 \geq (1-\epsilon)\Score_k(F^{-1}(\wtv), \PntSet)=
 (1-\epsilon)\Score_k(\wtv, M\PntSet)=(1-\epsilon)\Score_k(\wtv, \PntSet')$

 Conversely, if $\SbSet'$ is a $(k,\epsilon)$-regret set for $\PntSet'$, then for any
 $\wtv\in \SphereDP$,
 $\Score_1(\wtv,\SbSet)=\Score_1(\wtv,M^{-1}\SbSet')=\Score_1((M^{-1})^T\wtv,\SbSet')=\\ \Score_1(F(\wtv),\SbSet')
 \geq (1-\epsilon)\Score_k(F(\wtv),\PntSet') = (1-\epsilon)\Score_k(\wtv,M^{-1}\PntSet')
 =(1-\epsilon)\Score_k(\wtv,\PntSet)$.
This completes the proof.
\end{proof}

We now describe the transformation of the input points, which is a non-uniform scaling of
$\PntSet$. Specifically, for each $1\leq j\leq d$, let $m_j=\max_{\pnt_i\in P}\pnt_{ij}$ be
the maximum value of the $j$th coordinate among all points. Let $\Basis\subseteq \PntSet$ be
the subset of at most $d$ points, one per coordinate, corresponding to these $m_j$ values.
We refer to $\Basis$ as the basis of $\PntSet$, and let $\textsc{Basis}(\PntSet)$ be the method to find the basis $\Basis$.
We divide the $j$-th coordinate of all points by $m_j$, for all $j=1,2,\ldots,d$.
Let $\PntSet'$ be the resulting set,
and let $\Basis'$ be the transformation of $\Basis$. We note that for each coordinate
$j$ there is a point $\pnt_i'\in\Basis'$ with $\pnt_{ij}'=1$.
The different scaling factor in each coordinate can be represented by the diagonal matrix
$M$ where $M_{jj}=1/m_j$, and so $\PntSet'=M\PntSet$.
Let $\textsc{Scale}(\PntSet)$ be the procedure that scales the set $\PntSet$ according to the above transformation.
The key property of this affine transformation is the following lemma.
\begin{lemma}
\lemlab{Ltransf}
Let $M$ be the affine transform described above and let $\PntSet'=M\PntSet$.
Then, for all $\wtv\in \posPoints$,
\[
 \sqrt{d}\cdot \magn{\wtv} \:\geq \: \Score(\wtv,\PntSet')  \:\geq \: \frac{1}{\sqrt{d}}\cdot\magn{\wtv}.
\]
\end{lemma}
\begin{proof}
Since $\Score(\cdot,\cdot)$ is a linear function, without loss of generality consider a
vector $\wtv\in\posPoints$ with $\magn{\wtv}=1$.
 After the transformation $M$, for each coordinate $j$, we have $\pnt_{j}'\leq 1$.
 Therefore, $\magn{\pnt'}\leq\sqrt{d}$ and also $\sqrt{d}\geq \Score(\wtv,\PntSet')$
	because $\wtv$ is a unit vector.
For the second inequality, we note that for any unit norm vector $\wtv$ we must have
$\wtv_j\geq\frac{1}{\sqrt{d}}$, for some $j$. Since our transform ensures the existence
of a point $\pnt'\in \Basis'$ with $\pnt_{j}'=1$, we must have
$\Score(\wtv,\PntSet')\geq \dotp{\wtv}{\pnt'}\geq \frac{1}{\sqrt{d}}$.
This completes the proof.
\end{proof}

In the following, without loss of generality, we assume that $\PntSet\subset [0,1]^d$ and
there is a set $\Basis\subseteq \PntSet$ of at most $d$ points, such that
for any $1\leq j\leq d$, there is a point $\pnt\in\Basis$ with $\pnt_j=1$.

\subsection{Approximation Algorithms}
We first show how to formulate the \sizeProb{} version of the \Probl as a hitting set problem.
Let $\PntSet$, $k$, and $\epsilon$ be fixed. For a vector $\wtv\in \posPoints$, let
$\SetEps{\wtv} =\{\pnt \in \PntSet\mid \Score(\wtv,\pnt)\geq (1-\epsilon)\Score_k(\wtv,\pnt)\}$.
Note that if $\epsilon=0$, then $\SetEps{\wtv} =\kSet_k(\wtv)$, the set of top-$k$ points of $\PntSet$
in direction $\wtv$. Set $\AllSets=\{\SetEps{\wtv} \mid \wtv\in \posPoints\}$.
Although there are infinitely many preferences we show below that $\cardin{\AllSets}$ is polynomial
in $\cardin{\PntSet}$. We now define the set system $\Sigma=(\PntSet, \AllSets)$.

\begin{lemma}
\lemlab{lem:HS}
\begin{enumerate}[(i)]
\item $\cardin{\AllSets}=O(n^d)$.
\item A subset $\SbSet\subseteq \PntSet$ is a hitting set of $\Sigma$ if and only if $\SbSet$ is a $(k,\epsilon)$-regret set of $\PntSet$.
\end{enumerate}
\end{lemma}
\begin{proof}
(i) Note that $\SetEps{\wtv}$ is a subset of $\PntSet$ that is separated from $\PntSet\setminus\SetEps{\wtv}$ by
the hyperplane $\hpr_{\wtv}: \dotp{\wtv}{x}\geq (1-\epsilon)\Score_k(\wtv,\PntSet)$. Such a subset is called
linearly separable. A well-known result in discrete geometry \cite{agarwal2000arrangements} shows that
a set of $n$ points in $\Re^d$ has $O(n^d)$ linearly separable subsets. This completes the proof of (i).

(ii) First, by definition any $(k,\epsilon)$-regret set $\SbSet$ has to contain a point of $\SetEps{\wtv}$ for
all $\wtv\in \posPoints$ because otherwise $\ell_k(\wtv,\SbSet)>\epsilon$. Hence, $\SbSet$ is a hitting set
of $\Sigma$. Conversely, if $\SbSet\cap \SetEps{\wtv}\neq \emptyset$, then $\ell_k(\wtv,\SbSet)\leq \epsilon$.
If $\SbSet$ is a hitting set of $\Sigma$, then $\SbSet\cap \SetEps{\wtv}\neq \emptyset$ for all $\wtv\in \posPoints$,
so $\SbSet$ is also a $(k,\epsilon)$-regret set.
\end{proof}

We can thus compute a small-size $(k,\epsilon)$-regret set of $\PntSet$ by running the greedy hitting set
algorithm on $\Sigma$. In fact, the greedy algorithm in \cite{Bronniman1995} returns a hitting set of size
$O(\minSz\log \minSz)$. As mentioned above, the challenge is the size of $\AllSets$. Even for small values
of $k$, $\cardin{\AllSets}$ can be $\Omega(n^{\floor{d/2}})$ \cite{agarwal2000arrangements}.
Next, we show how to construct a much smaller set system.

Recall that $\ell_k(\wtv,\SbSet)$ is independent of $\magn{\wtv}$ so we focus on unit preference vectors, i.e.,
we assume $\magn{\wtv}=1$. Let $\AllDir=\{\wtv\in\posPoints\mid\magn{\wtv}=1\}$ be the space of all unit
preference vectors; $\AllDir$ is the portion of the unit sphere restricted to the positive orthant.
For a given parameter $\delta>0$, a set $\Net\subset \AllDir$ is called a \emph{$\delta$-net} if the spherical
caps of radius $\delta$ around the points of $\Net$ cover $\AllDir$, i.e. for any $\wtv\in \AllDir$, there is
a point $v\in \Net$ with $\dotp{\wtv}{v}\geq \cos(\delta)$. A $\delta$-net of size $O(\frac{1}{\delta^{d-1}})$
can be computed by drawing a "uniform" grid on $\AllDir$. In practice, it is simpler and more efficient
to simply choose a random set of $O(\frac{1}{\delta^{d-1}}\log\frac{1}{\delta})$ directions --- this will be a $\delta$-net with probability at least $1/2$.
Set $\delta=\frac{\epsilon}{2d}$. Let $\Net$ be be a $\delta$-net of $\AllDir$, and let
$\AllSetsNet=\{\SetEps{\wtv}\mid \wtv\in \Net\}$.

Set $\Sigma_{\Net}=(\PntSet,\AllSetsNet)$.
Note that $\cardin{\AllSetsNet}=O(\frac{1}{\epsilon^{d-1}})$. Our main observation, stated in the lemma below,
is that it suffices to compute a hitting set of $\Sigma_{\Net}$. That is, a subset $\SbSet\subseteq \PntSet$
that has a small regret with respect to vectors in $\Net$ has small regret for all preferences.

\begin{figure*}
\begin{subfigure}{0.25\textwidth}
  \includegraphics[width=\linewidth]{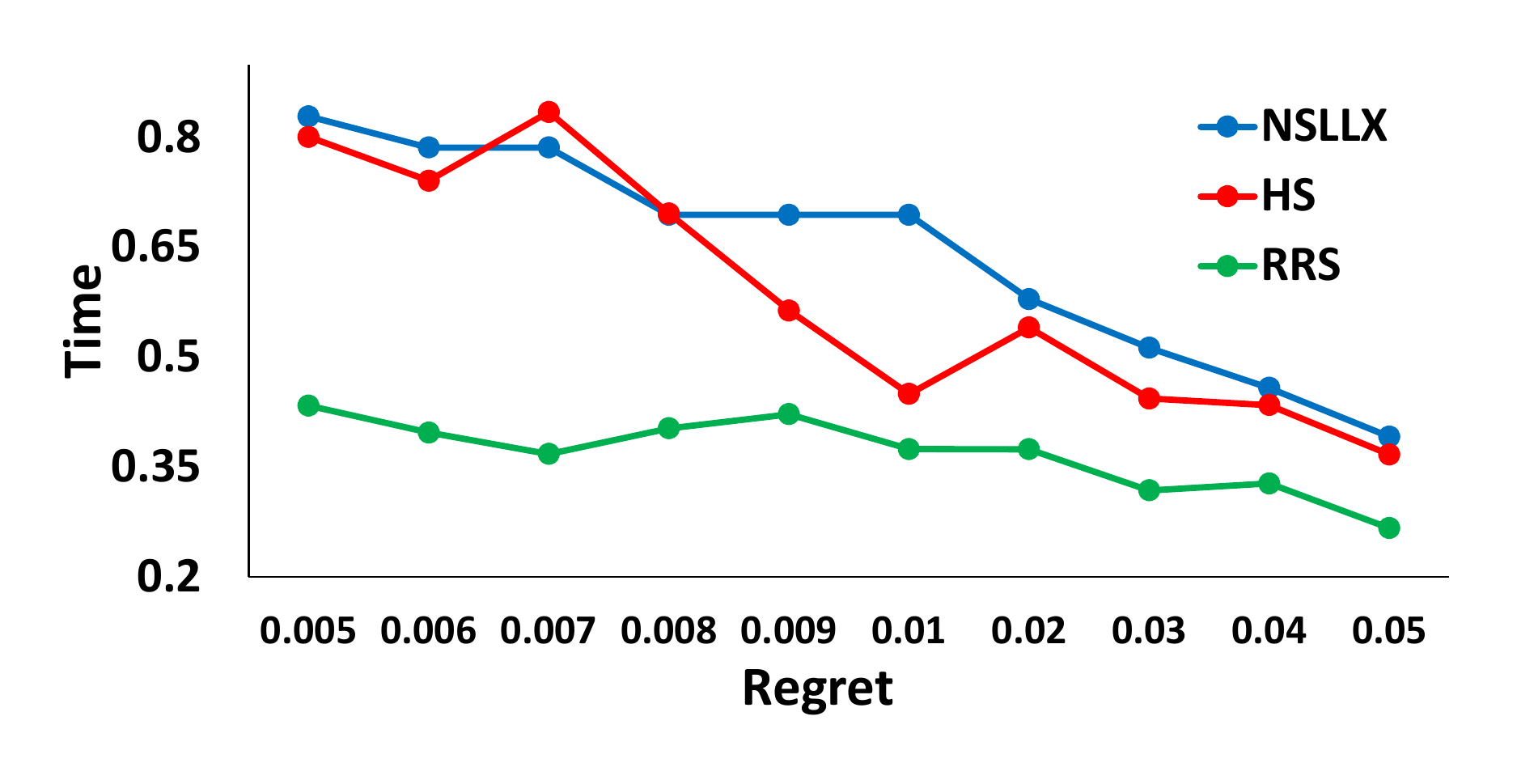}
  \caption{BB}\label{fig:BBT1}
\end{subfigure}\hfill
\begin{subfigure}{0.25\textwidth}
  \includegraphics[width=\linewidth]{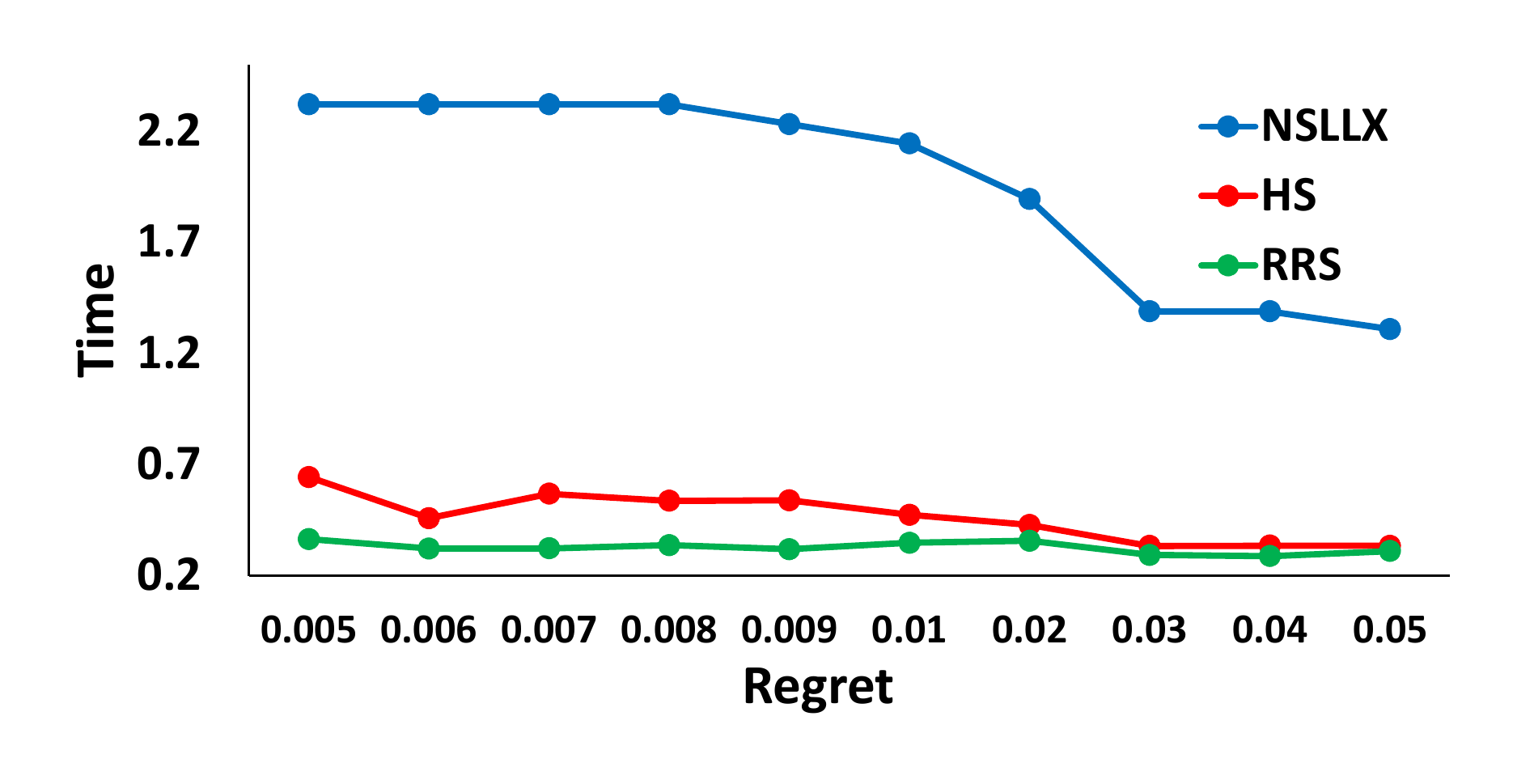}
  \caption{AntiCor}\label{fig:ACT1}
\end{subfigure}\hfill
\begin{subfigure}{0.25\textwidth}
  \includegraphics[width=\linewidth]{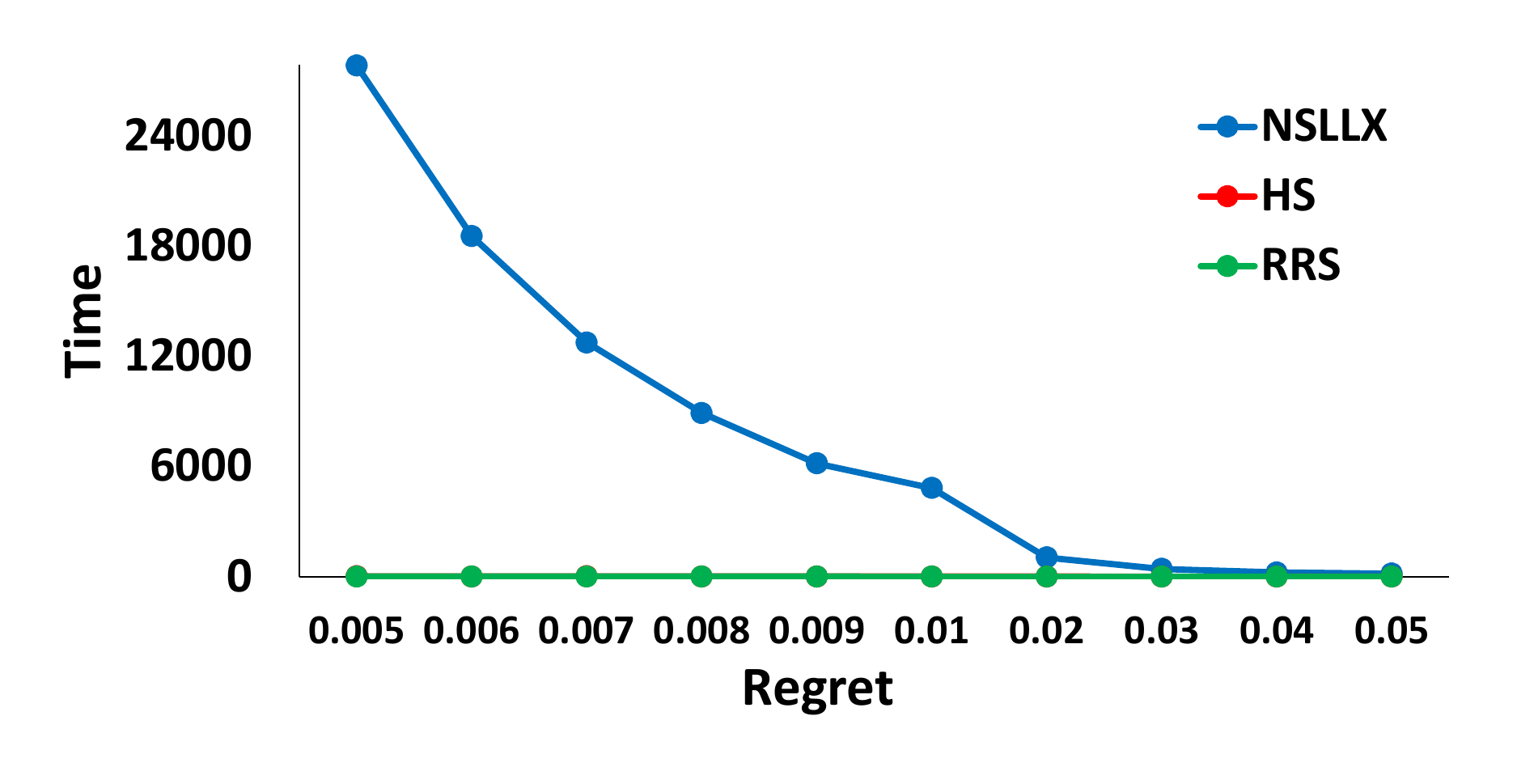}
  \caption{Sphere}\label{fig:PST1}
\end{subfigure}\hfill
\begin{subfigure}{0.25\textwidth}
  \includegraphics[width=\linewidth]{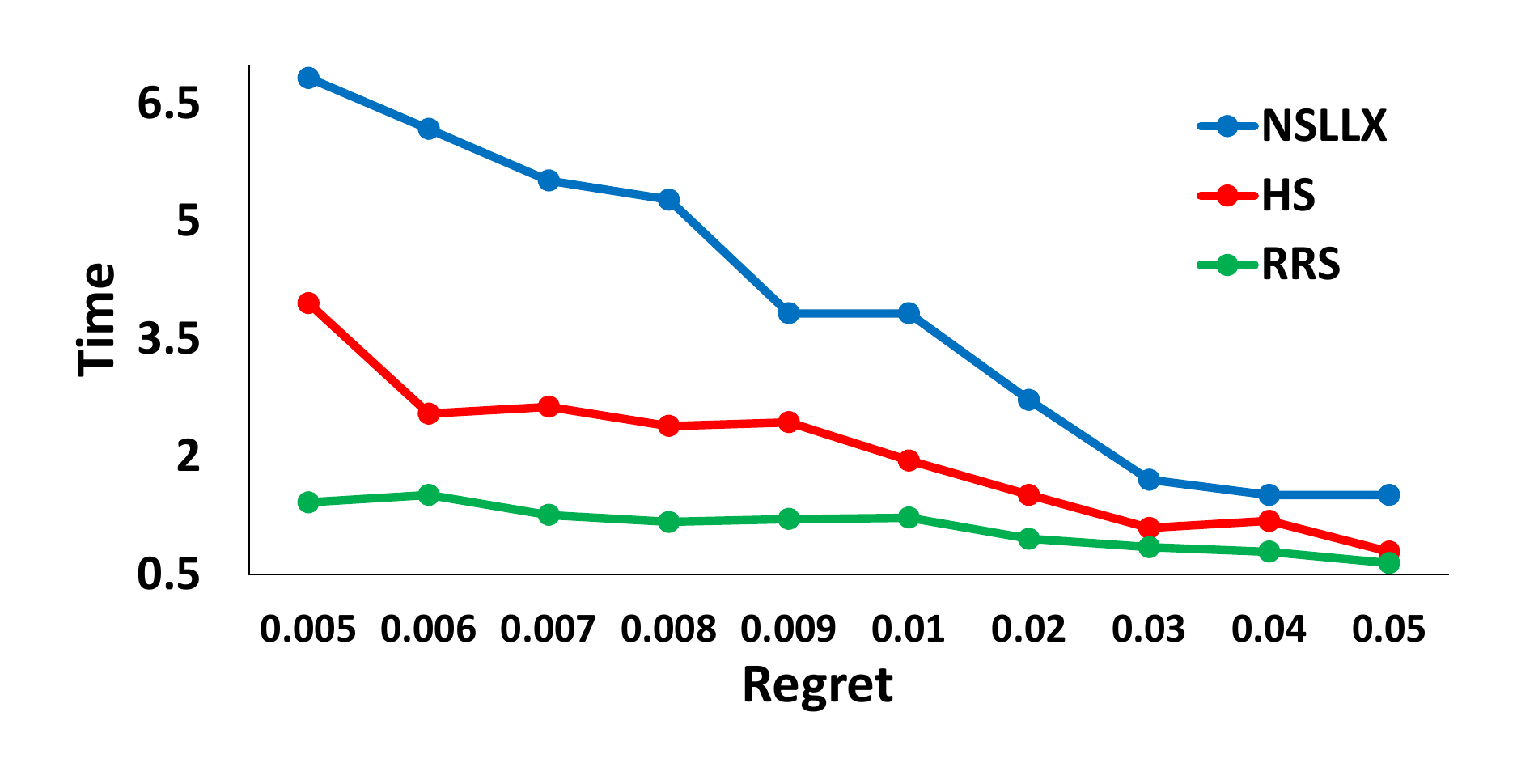}
  \caption{ElNino}\label{fig:ENT1}
\end{subfigure}
\caption{Running time for $k=1$.}
\figlab{fig:TimeK1}
\end{figure*}

\begin{figure*}
\begin{subfigure}{0.25\textwidth}
  \includegraphics[width=\linewidth]{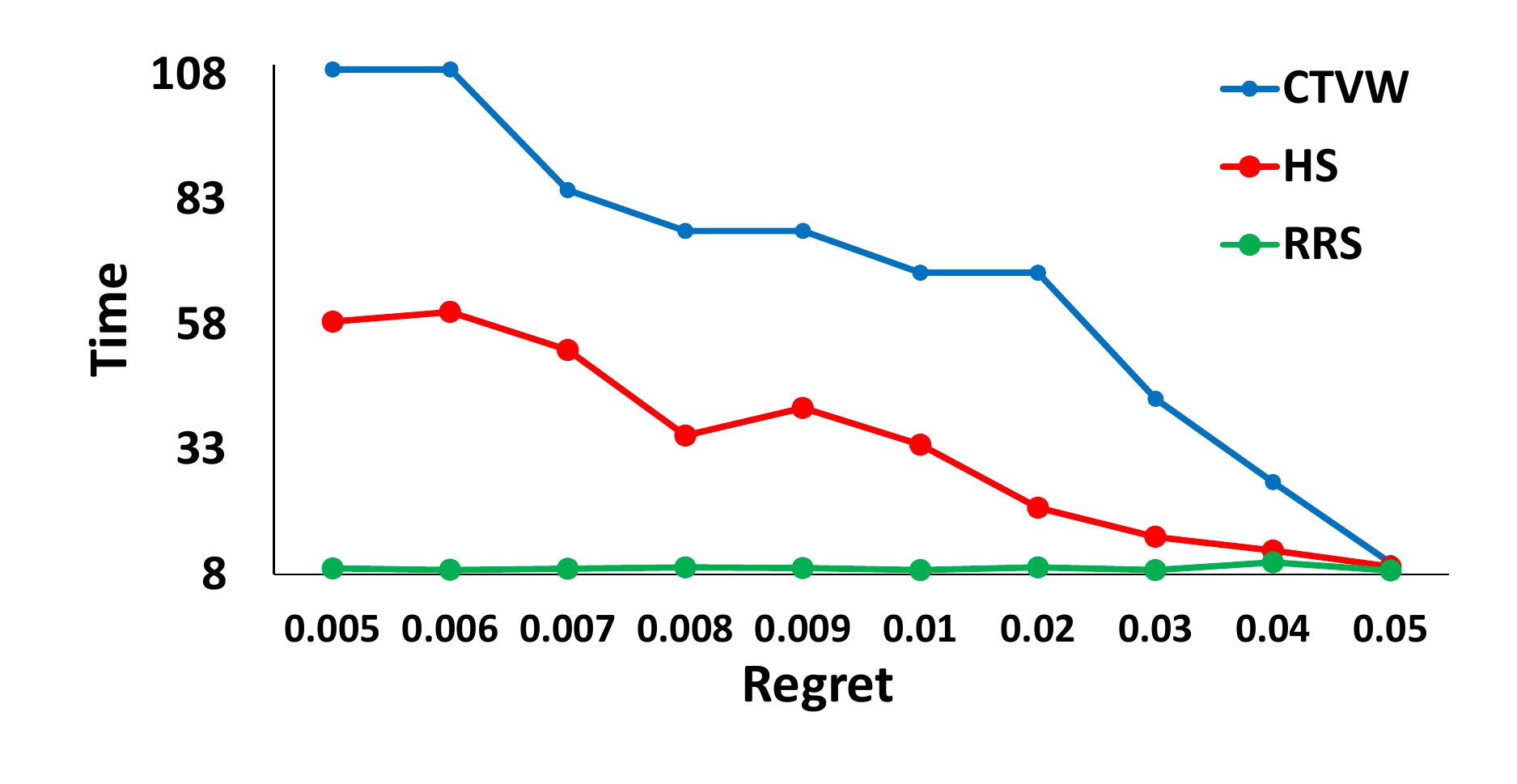}
  \caption{BB}\label{fig:BBT10}
\end{subfigure}\hfill
\begin{subfigure}{0.25\textwidth}
  \includegraphics[width=\linewidth]{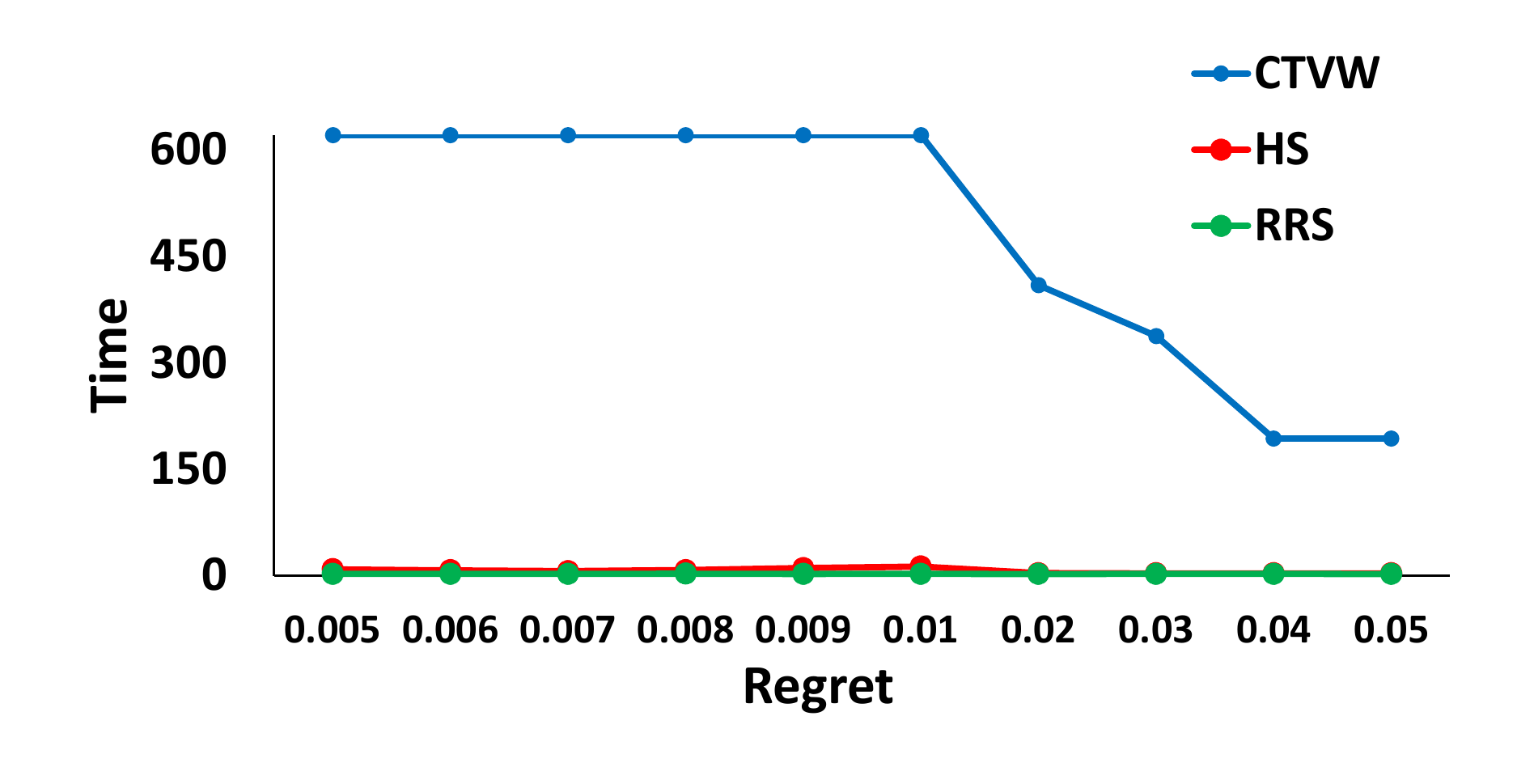}
  \caption{AntiCor}\label{fig:ACT10}
\end{subfigure}\hfill
\begin{subfigure}{0.25\textwidth}
  \includegraphics[width=\linewidth]{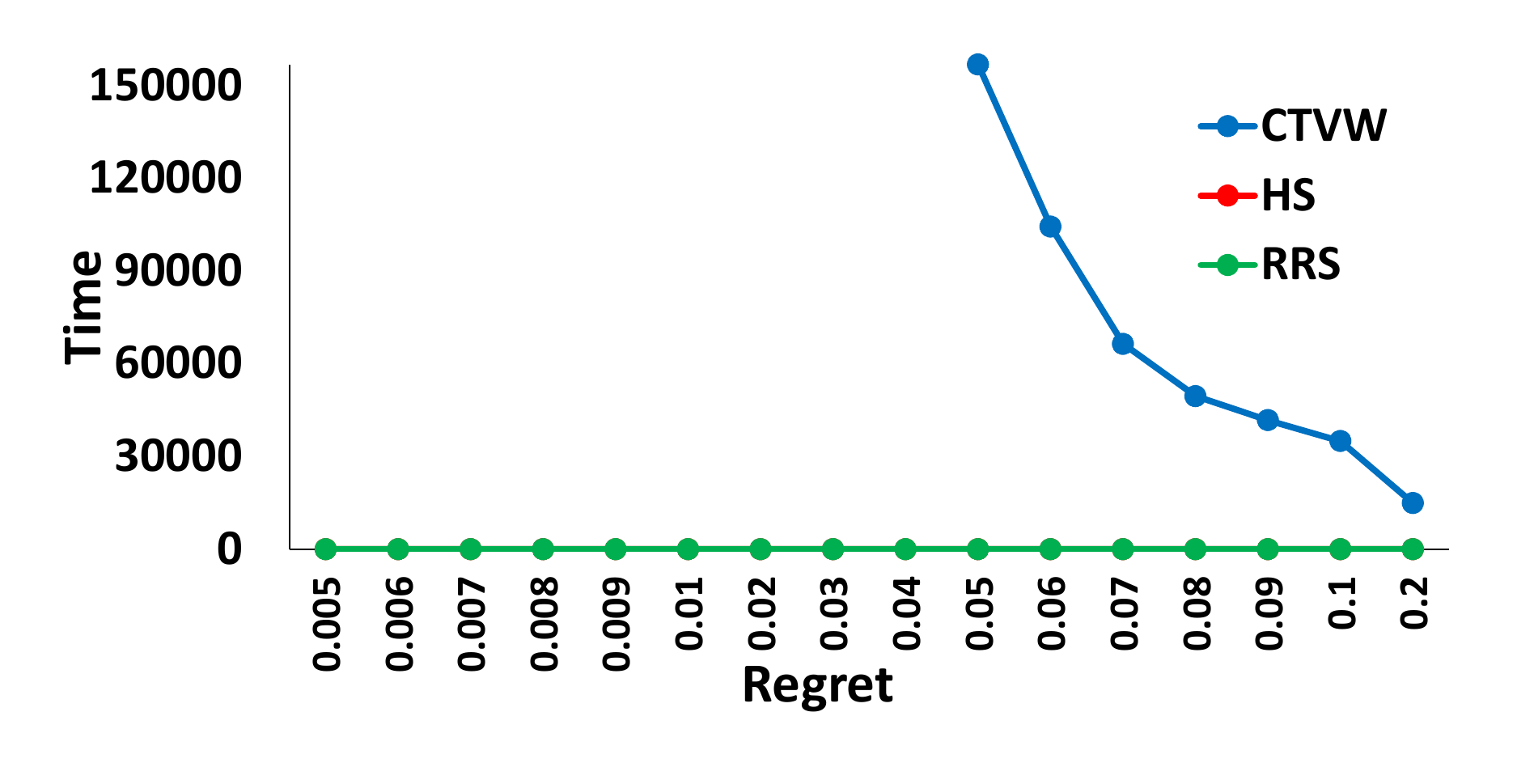}
  \caption{Sphere}\label{fig:PST10}
\end{subfigure}\hfill
\begin{subfigure}{0.25\textwidth}
  \includegraphics[width=\linewidth]{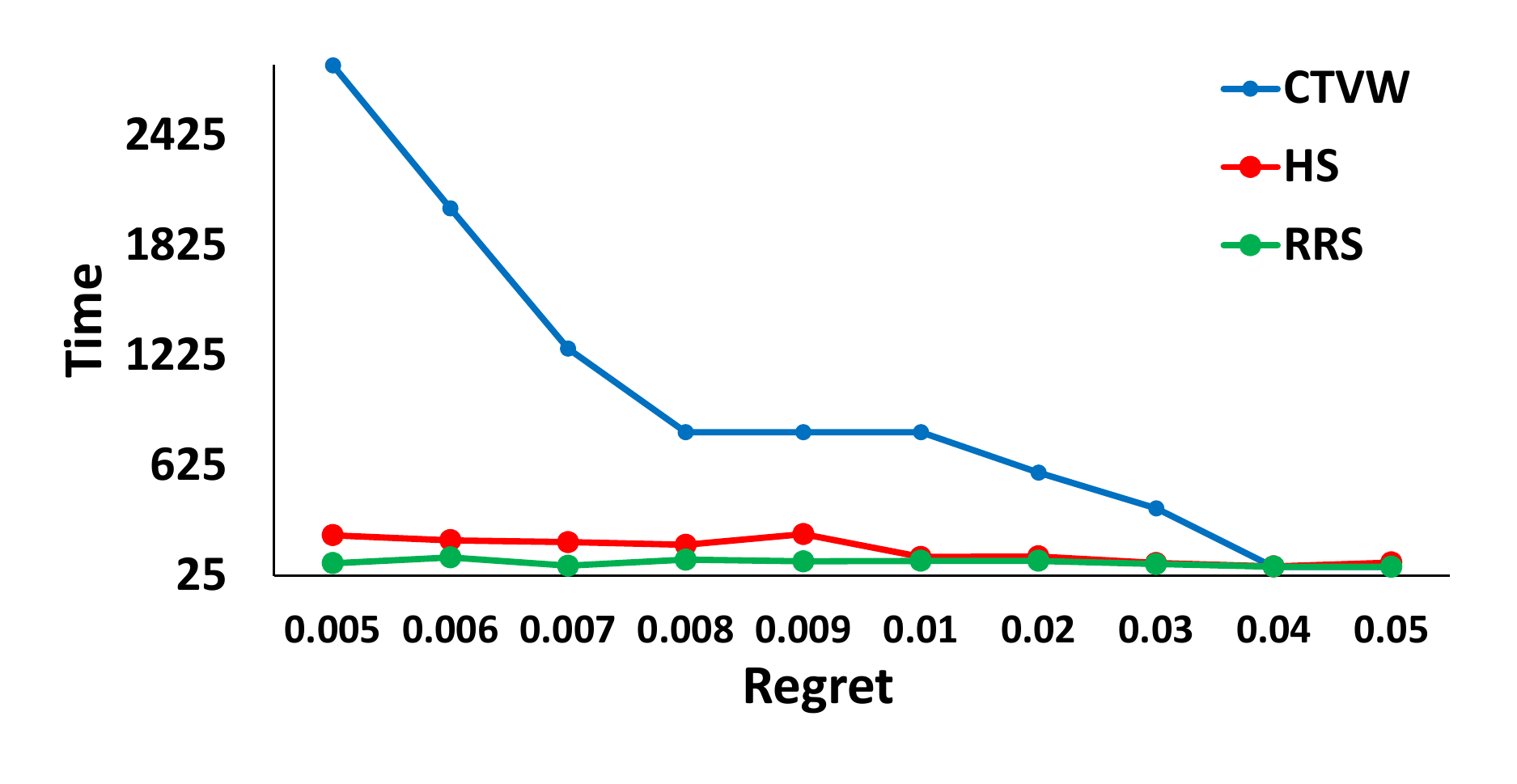}
  \caption{ElNino}\label{fig:ENT10}
\end{subfigure}
\caption{Running time for $k=10$.}
\figlab{fig:TimeK10}
\end{figure*}

\begin{lemma}
\lemlab{MainLemma}
Let $\SbSet'$ be a hitting set of $\Sigma_{\Net}$, and let $\Basis$ be the basis of $\PntSet$.
Then $\SbSet=\SbSet'\cup \Basis$ is a $(k,2\epsilon)$-regret set of $\PntSet$.
\end{lemma}


%

\algoref{HitSetAlg} summarizes the algorithm. $\textsc{Greedy\_HS}$ is the greedy algorithm in \cite{Bronniman1995}
for computing a hitting set.
\captionof{algorithm}{RMS\_HS}\algolab{HitSetAlg}
\noindent Input: $\PntSet$: Input points, $k\geq 1$: rank, $\epsilon>0$: error parameter.\\
\noindent Output: $\SbSet$ a $(k,2\epsilon)$-regret set
\begin{algorithmic}[1]
\State $B:=\textsc{Basis}(\PntSet)$
\State $\PntSet:= \textsc{Scale}(\PntSet)$.
\State $\delta:= \frac{\epsilon}{2d}$
\State $\Net:= \delta$-net of $\AllDir$
\State $\SetEps{\wtv} := \{\pnt\in \PntSet\mid \Score(\wtv,\pnt)\geq (1-\epsilon)\Score_k(\wtv,\PntSet)\}$
\State $\AllSetsNet:= \{\SetEps{\wtv} \mid \wtv\in \Net\}$
\State $\SbSet':= \textsc{Greedy\_HS}(\PntSet, \AllSetsNet)$
\State Return $\SbSet:= \SbSet'\cup \Basis$
\end{algorithmic}

\vspace{2 mm}

\mparagraph{Analysis}
The correctness of the algorithm follows from \lemref{MainLemma}.
Since a hitting set of $\Sigma$ is also a hitting set of $\Sigma_{\Net}$, $\Sigma_{\Net}$ has
a hitting set of size at most $\minSz$. The greedy algorithm in \cite{Bronniman1995} returns
a hitting set of size $O(\minSz\log \minSz)$ for $d\geq 4$ and of size $O(\minSz)$ for $d\leq 3$.
Therefore $\cardin{\SbSet}=O(\minSz\log \minSz)$ for $d\geq 4$ and $O(\minSz)$ for $d=3$.
Computing the set $\Basis$ takes $O(n)$ time. $\Net$ can be constructed in $O(\cardin{\Net})$ time
and we can compute $\SetEps{\wtv}$ for each $\wtv\in \Net$ in $O(n)$ time. The greedy algorithm in
\cite{Bronniman1995} takes $O(\frac{n}{\epsilon^{d-1}}\log n\log\frac{1}{\epsilon})$ expected time
(the bound on the running time also holds with high probability).

We now prove \lemref{MainLemma}.
\mparagraph{Proof of \lemref{MainLemma}}
\newcommand{\sd}{\bar{\wtv}}
It suffices to show that for any direction $\wtv\in \AllDir$ there is a point $q\in \SbSet$ for which
$\Score(\wtv,q)\geq (1-2\epsilon)\Score_k(\wtv,\PntSet)$.

Let us first consider the case when $\Score_k(\wtv,\PntSet) \leq \frac{1}{(1-\epsilon)\sqrt{d}}$.
In this case, by \lemref{Ltransf} the set $\Basis$ is guaranteed to contain a point $q$ with
$\Score(\wtv,q) \geq \frac{1}{\sqrt{d}}$, which proves the claim.
So let us now assume that $\Score_k(\wtv,\PntSet) > \frac{1}{(1-\epsilon)\sqrt{d}}$.
Let $\sd\in \Net$ be a direction in the net $\Net$ such that, $(\widehat{\wtv,\sd}) \leq \epsilon/2d$,
where $(\widehat{\wtv,\sd})$ is the angle between $\wtv$ and $\sd$.
Such a direction exists because $\Net$ is a $\frac{\epsilon}{2d}$-net on $\AllDir$. Observe that,
	\[
		\norm{\wtv - \sd} = \sqrt{2 - 2 \cos((\widehat{\wtv,\sd}))}
	                      = 2 \sin\left(\frac{(\widehat{\wtv,\sd})}{2}\right)
			      \leq \frac{\epsilon}{2d},
	\]
	where we have used first the cosine rule, the identity
	$1 - \cos \theta = 2 \sin^2 \left( \frac{\theta}{2} \right)$,
	as well as the inequality $\sin \theta \leq \theta$ for $\theta \geq 0$ in the final step. Also, observe that for any $\pnt \in \PntSet$ we have,
	\begin{equation}
		\Eqlab{eqnA}%
		\abs{\Score(\wtv,\pnt) - \Score(\sd,\pnt)} \leq \frac{\epsilon}{2\sqrt{d}}.
	\end{equation}
	This follows because,
	\begin{align*}
	  \abs{\Score(\wtv,\pnt) - \Score(\sd,\pnt)} & = \abs{\dotp{\wtv}{\pnt} - \dotp{\sd}{\pnt}} = \abs{\dotp{\wtv - \sd}{\pnt}}\\
	                                               & \leq \norm{\wtv - \sd}\norm{\pnt} \leq \frac{\epsilon}{2d} \times \sqrt{d} = \frac{\epsilon}{2\sqrt{d}},
	\end{align*}
	where we have used the Cauchy-Schwarz inequality for the first
	inequality, the upper bound on $\norm{\wtv - \sd}$ derived earlier, along with $\norm{\pnt} \leq \sqrt{d}$ for the second inequality.

	Let $x_1, x_2, \ldots, x_k \in \PntSet$ be the top-$k$ points along direction $\wtv$, i.e., $x_i = \kPoint_i(\wtv,\PntSet)$. Also, let $y_k$ be the top-$k$
	point along direction $\bar{\wtv}$. As remarked we can assume,
	$\Score(\wtv,x_i) \geq \Score(\wtv,x_k) \geq \frac{1}{(1-\epsilon)\sqrt{d}}$.
	Now, for any $i = 1, 2, \ldots, k$ we have that,
	\begin{align*}
	\Score(\sd,x_i) & \geq \Score(\wtv,x_i) - \frac{\epsilon}{2\sqrt{d}}\\
									 & \geq \Score(\wtv,x_i) - \frac{(1-\epsilon)\epsilon}{2} \Score(\wtv,x_i)\\
	       &= \Score(\wtv,x_i)\left(1 - \frac{(1-\epsilon)\epsilon}{2}\right)\\
	       & \geq \Score(\wtv,x_k)\left(1 - \frac{(1-\epsilon)\epsilon}{2}\right)   .
	\end{align*}
    The first inequality follows by $\Eqref{eqnA}$, and the second inequality holds since
    $\Score(\wtv,x_i) \geq \Score(\wtv,x_k) > \frac{1}{(1-\epsilon)\sqrt{d}}$.
	This implies that there are $k$ points whose scores are each at least
	$\Score(\wtv,x_k)\left(1 - \frac{(1-\epsilon)\epsilon}{2} \right)$, and therefore the $k$-th best score along $\sd$, i.e., $\Score(\sd,y_k)$, is at least $\Score(\wtv,x_k)\left(1 - \frac{(1-\epsilon)\epsilon}{2}\right)$. Now, the algorithm guarantees that there is a point $q \in Q$ such that
	$\Score(\bar{\wtv},q) \geq (1-\epsilon)\Score(\bar{\wtv},y_k)$. We claim
	that this $q$ ``settles'' direction $\wtv$ as well, up-to the factor $(1-\epsilon)^2$. Indeed,
	\begin{align*}
	\Score(\wtv,q) &\geq \Score(\bar{\wtv},q) - \frac{\epsilon}{2\sqrt{d}} \geq (1-\epsilon)\Score(\bar{\wtv},y_k) - \frac{\epsilon}{2\sqrt{d}} \\
		       &\geq (1-\epsilon)\left(1 - \frac{(1-\epsilon)\epsilon}{2}\right) \Score(\wtv,x_k) - \frac{\epsilon}{2\sqrt{d}}\\
		       &\geq (1-\epsilon)\left(1 - \frac{(1-\epsilon)\epsilon}{2}\right)\Score(\wtv,x_k) - \frac{(1-\epsilon)\epsilon}{2}\Score(\wtv,x_k) \\
		       &= (1-\epsilon)(1-\epsilon + \epsilon^2/2)\Score(\wtv,x_k) \geq (1-\epsilon)^2\Score(\wtv,x_k)\\
                &\geq (1-2\epsilon)\Score(\wtv,x_k).
	\end{align*}
	This completes the proof.

Putting everything together, we obtain the following:
\begin{theorem}
\label{Theor1}
Let $\PntSet\subset \posPoints$ be a set of $n$ points in $\Re^d$, $k\geq 1$ an integer,
and $\epsilon>0$ a parameter. Let $\minSz$ be the minimum size of a $(k,\epsilon)$-regret set
of $\PntSet$. A subset $\SbSet\subseteq \PntSet$ can be computed in $O\left(\frac{n}{\epsilon^{d-1}}\log (n)\log\left(\frac{1}{\epsilon}\right)\right)$
expected time such that $\SbSet$ is a $(k,2\epsilon)$-regret set of $\PntSet$. The size of $\SbSet$
is $O(\minSz\log \minSz)$ for $d\geq 4$ and $O(\minSz)$ for $d\leq 3$.
\end{theorem}

\mparagraph{Min-error RMS}
Recall that the \errorProb{} problem takes as input a parameter $r$, and returns a subset
$\SbSet\subseteq \PntSet$ of size at most $r$ such that $\ell_k(\SbSet) = \ell(r)$.
We propose a bicriteria approximation algorithm for the \errorProb{} problem by using
\algoref{HitSetAlg} for the \sizeProb{} version of the problem.
Let $c$ be a sufficiently large constant.
We perform a binary search on the values of the regret ratio $\bar{\epsilon}$ in the
range $[0,1]$. For each value of the $\bar{\epsilon}$ we run \algoref{HitSetAlg} with parameter $\epsilon=\bar{\epsilon}$ and let $\SbSet_{\bar{\epsilon}}$
be the returned set. If $\cardin{\SbSet_{\epsilon}}>cr\log r$ we set $\bar{\epsilon}\leftarrow 2\bar{\epsilon}$,
otherwise $\bar{\epsilon}\leftarrow \bar{\epsilon}/2$ and we continue the binary search with the new parameter.
We stop when we find a set $\SbSet_{\hat{\epsilon}}$ such that $\cardin{\SbSet_{\hat{\epsilon}}}<cr\log r$ and
$\cardin{\SbSet_{\hat{\epsilon}/2}}>cr\log r$. The stopping condition satisfies that
$\ell(cr\log r)\leq \hat{\epsilon}\leq 2\ell(r)$.
The following theorem summarizes the results of the \errorProb{} version of the problem.
\begin{theorem}
%
Let $\PntSet\subset \posPoints$ be a set of $n$ points in $\Re^d$, $k\geq 1$ an integer,
and $r>0$ a parameter. Let $\minError{r}$ be the minimum regret ratio of a subset of
$\PntSet$ of size at most $r$. A subset $\SbSet\subseteq \PntSet$ can be computed in
$O\left(\min\{\frac{n}{\ell_k(\SbSet)^{d-1}}\log (n)\log\left(\frac{1}{\ell_k(\SbSet)}\right),n^d\}\right)$ expected time such that
$\minError{cr\log r}\leq \ell_k(\SbSet)\leq 2\minError{r}$ for $d\geq 4$ and
$\minError{cr}\leq \ell_k(\SbSet)\leq 2\minError{r}$ for $d\leq 3$ for a sufficiently large constant $c$.
The size of $\SbSet$ is $O(r\log r)$ for $d\geq 4$ and $O(r)$ for $d\leq 3$.
\end{theorem}

\mparagraph{Remarks}
\begin{enumerate}[(i)]
\item For $k=1$ the optimum solution of the \Probl will always be a subset of the skyline of $\PntSet$.
Hence, to reduce the running time we can only run our algorithms on skyline points.
We can show that we still get the same approximation factors.
\item Instead of choosing $O(\frac{1}{\epsilon^{d-1}})$ directions in one step and find a Hitting Set,
we can sample in stages and maintain a hitting set until we find a $(k,\epsilon)$-regret set.
\end{enumerate} 
\section{Experiments} \seclab{experiments}

\begin{figure*}
\begin{subfigure}{0.25\textwidth}
  \includegraphics[width=\linewidth]{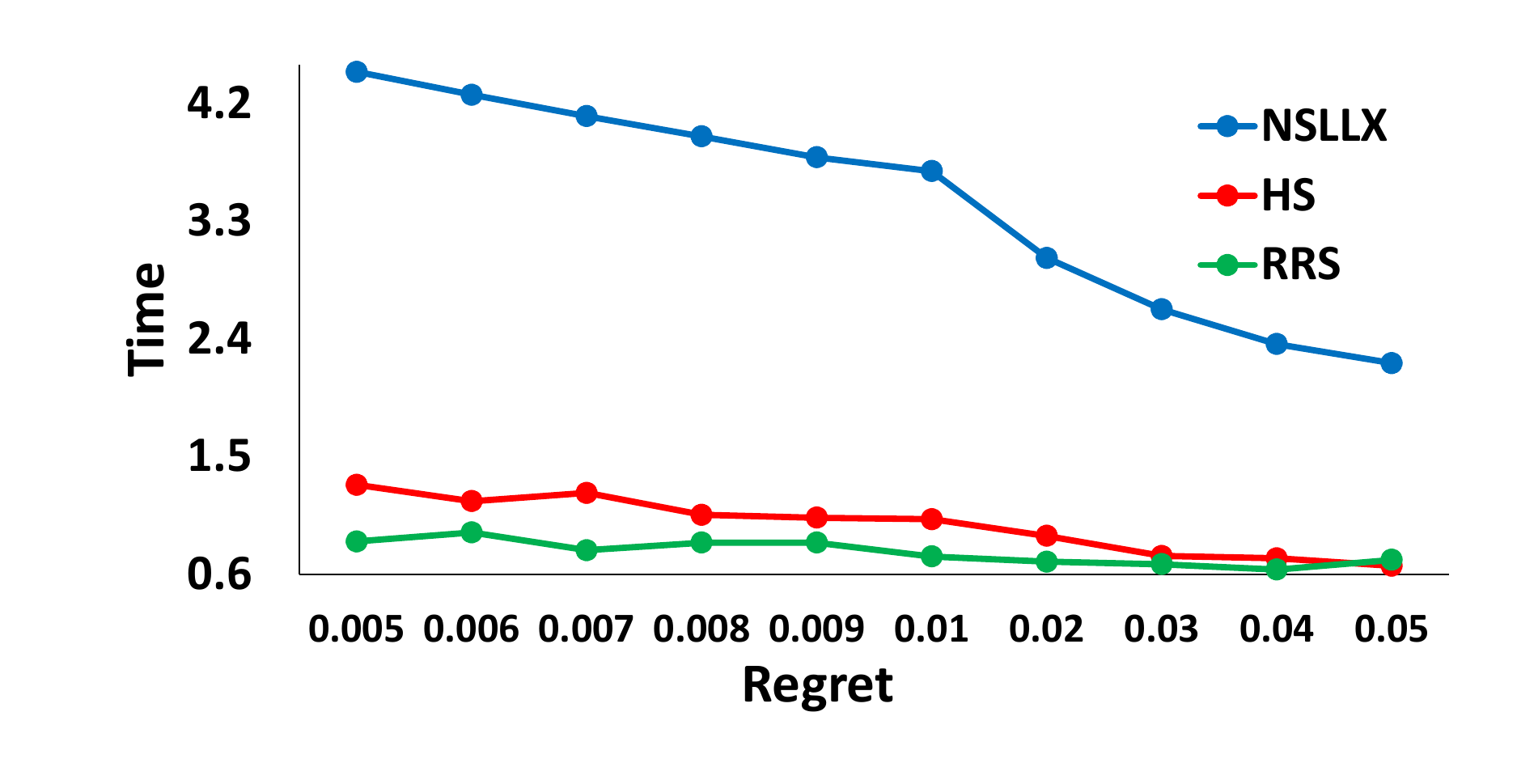}
  \caption{Sphere, $k=1$}\label{fig:LOGSD4_1}
\end{subfigure}\hfill
\begin{subfigure}{0.25\textwidth}
  \includegraphics[width=\linewidth]{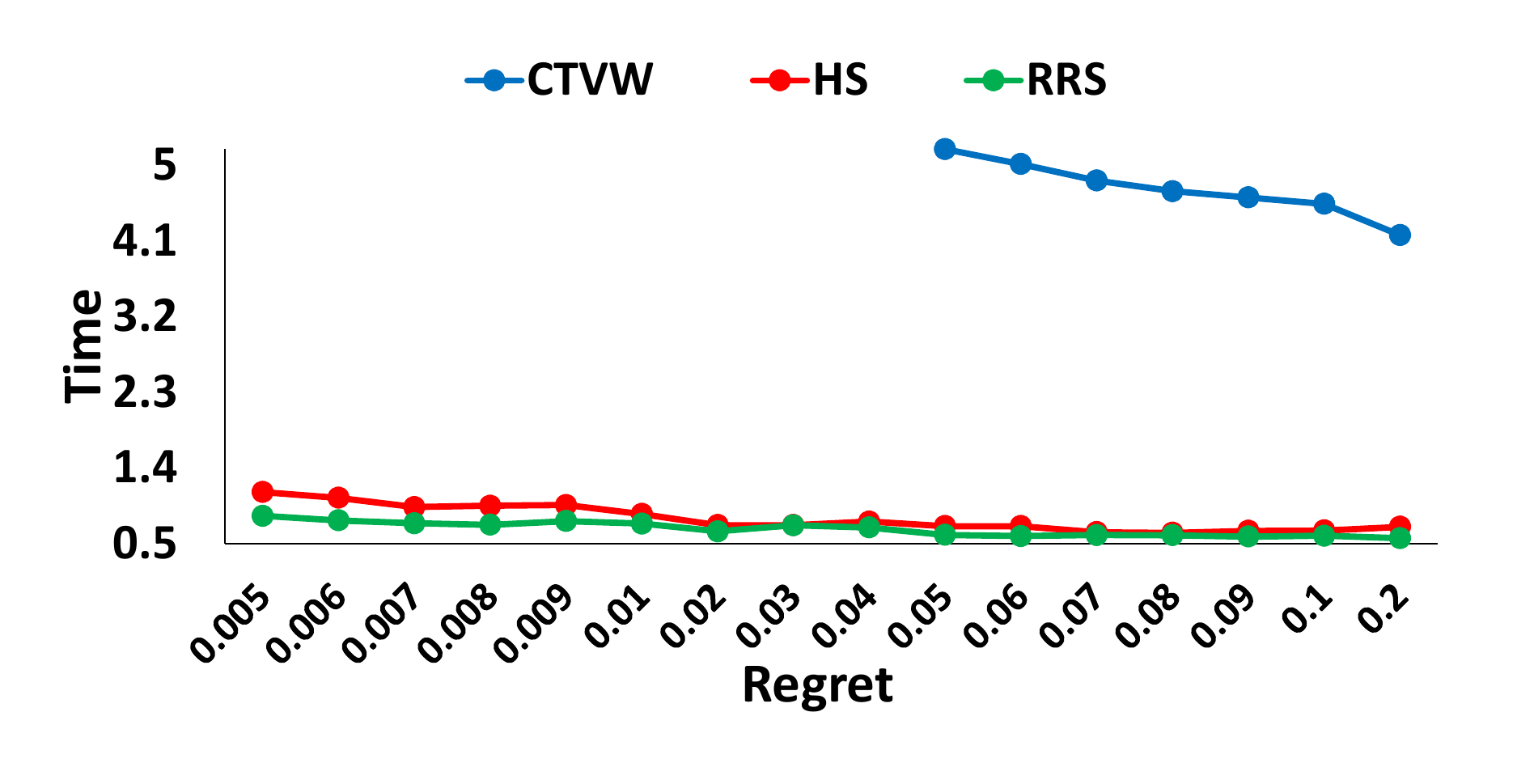}
  \caption{Sphere, $k=10$}\label{fig:LOGSD4_1}
\end{subfigure}\hfill
\begin{subfigure}{0.25\textwidth}
  \includegraphics[width=\linewidth]{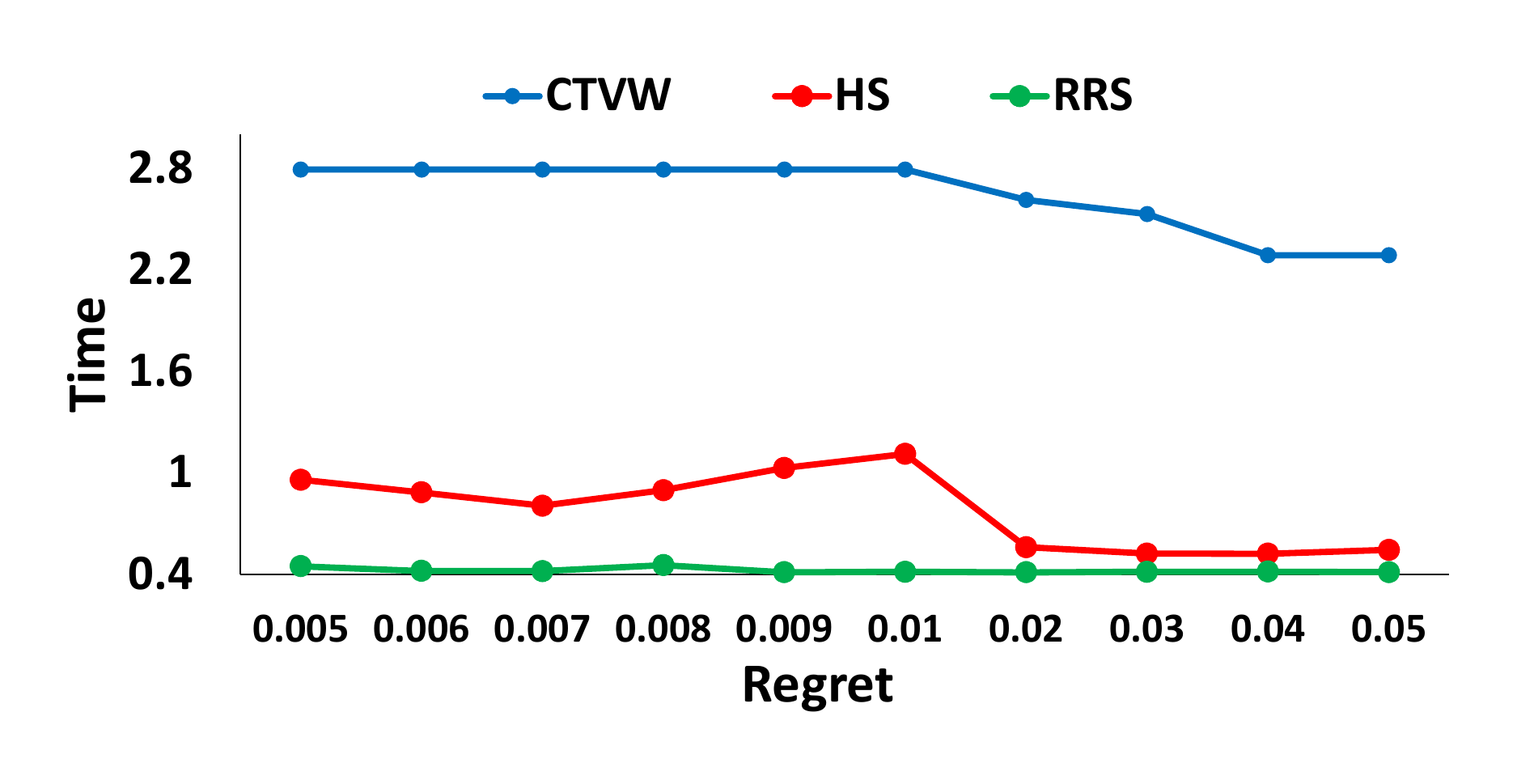}
  \caption{AntiCor, $\sigma=0.1$, $k=10$}\label{fig:LOGAntiCor_10}
\end{subfigure}\hfill
\begin{subfigure}{0.25\textwidth}
  \includegraphics[width=\linewidth]{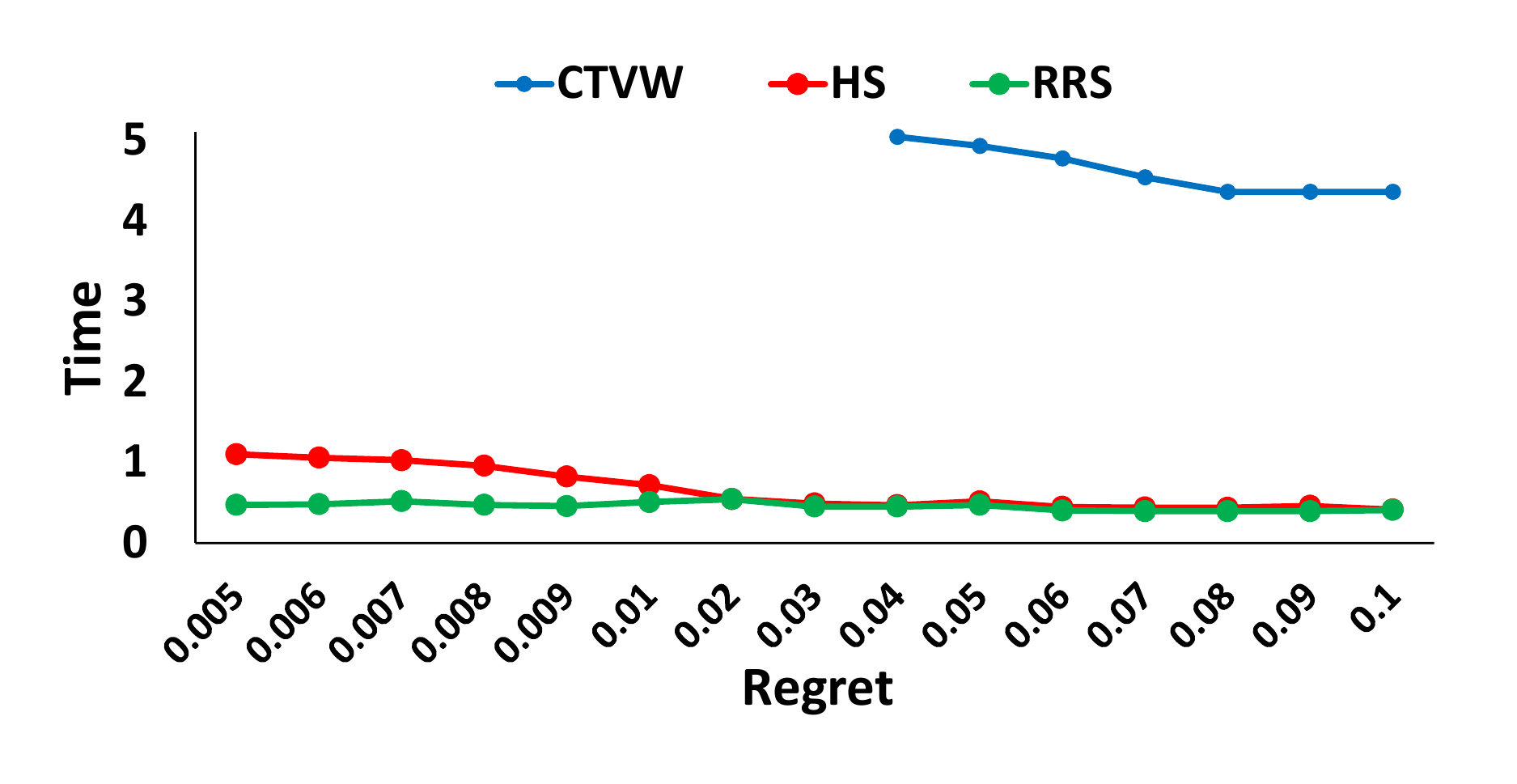}
  \caption{AntiCor, $\sigma=0.01$, $k=10$}\label{fig:LOGAntiCors001_10}
\end{subfigure}\hfill
\caption{$\log_{10}$-scale running time.}
\label{fig:LOGTimeSPhere}
\end{figure*}

\begin{figure}
\begin{subfigure}{0.24\textwidth}
  \includegraphics[width=\linewidth]{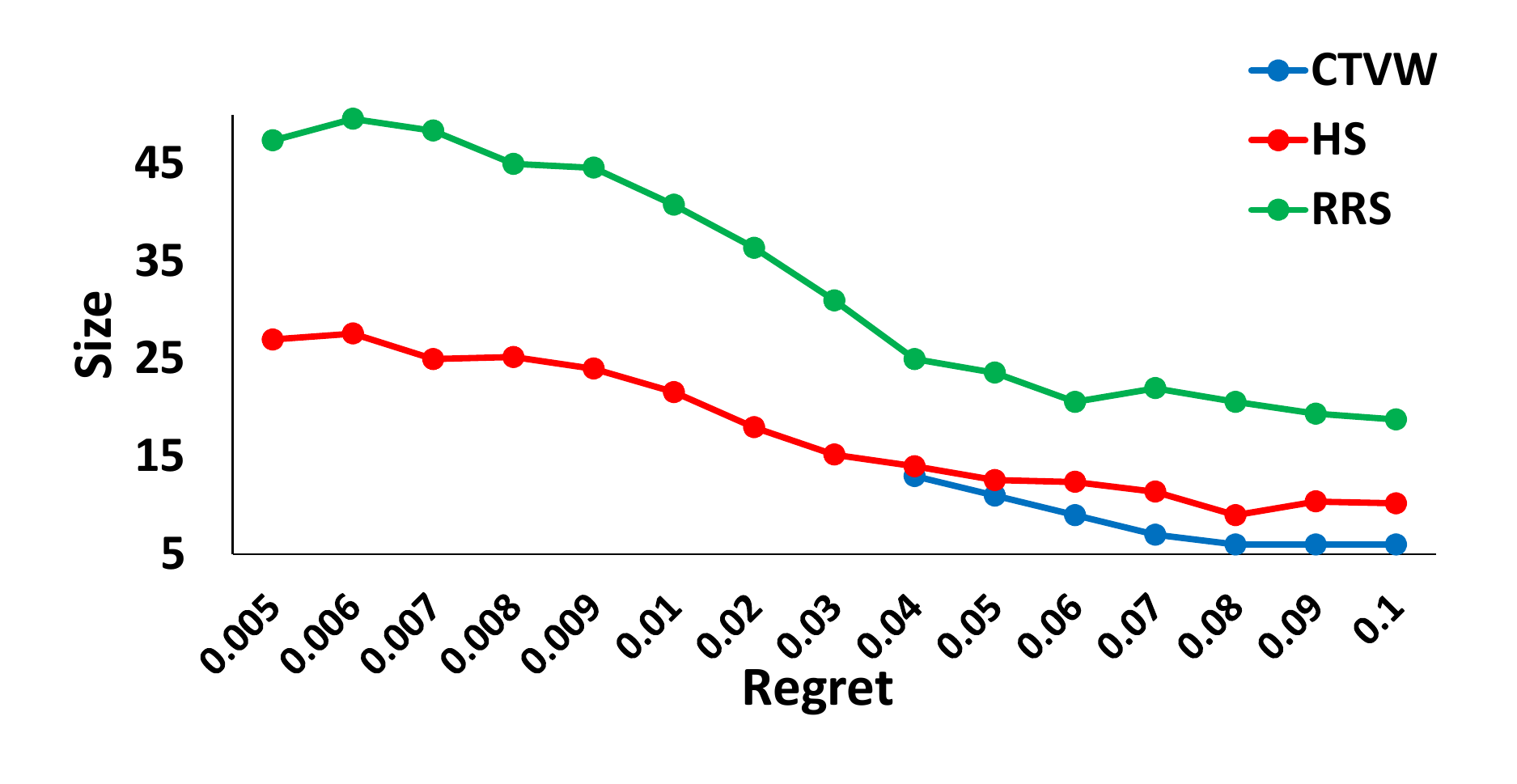}
  \caption{Regret Ratio}\label{fig:AntiCors001K10}
\end{subfigure}\hfill
\begin{subfigure}{0.24\textwidth}
  \includegraphics[width=\linewidth]{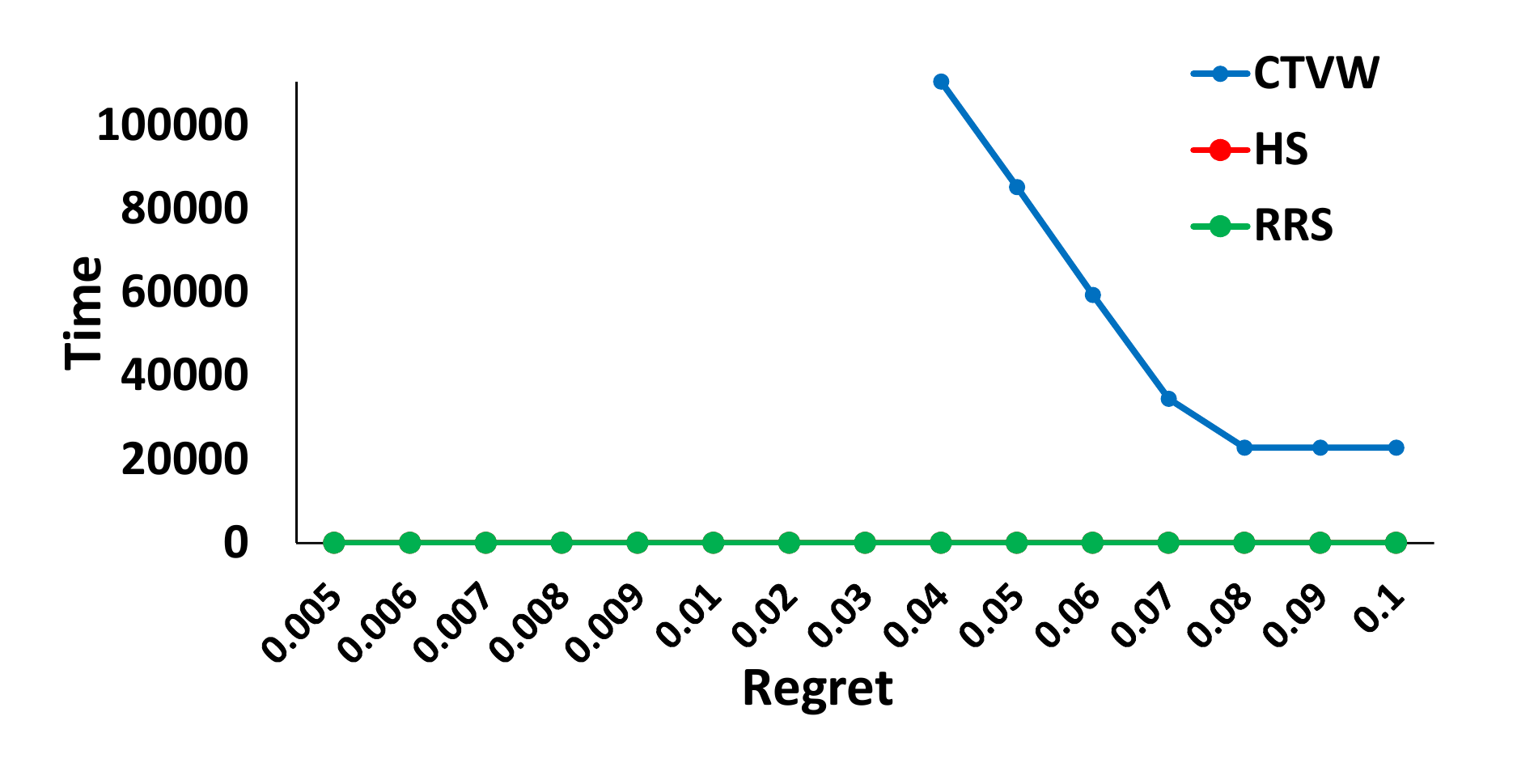}
  \caption{Running Time}\label{fig:AntiCors001TimeK10}
\end{subfigure}\
\caption{Regret ratio and running time of AntiCor, $\sigma=0.01$, $k=10$.}
\label{fig:AntiCorD4s001}
\end{figure}

We have implemented our algorithms as well as the current state of the art, namely, the greedy
algorithms described in~\cite{nanongkai2010regret, chester2014computing}, and experimentally
evaluated their relative performance.

\mparagraph{Algorithms}
In particular, the four algorithms we evaluate are the following:

  \textbf{\textit{\CoreSet{}}} is the \textbf{\textit{R}}andomized \textbf{\textit{R}}egret \textbf{\textit{S}}et
	algorithm, based on coresets, described in \secref{coreSet}. In our implementation,
	rather than choosing $O(\frac{1}{\epsilon^{(d-1)/2}})$ random preferences all at once,
	we choose them in stages and maintain a subset $\SbSet$ until $\ell_k(\SbSet)\leq \epsilon$.

\smallskip

  \textbf{\textit{\HS{}}} is the \textbf{\textit{H}}itting \textbf{\textit{S}}et algorithm
	presented in \secref{approxAlg}, and our implementation incorporates the remarks made at the end of \secref{approxAlg}.

\smallskip

  \textbf{\textit{\GreedyKI{}}} is the greedy algorithm for $1$-\Probl described in \cite{nanongkai2010regret},
	which iteratively finds the preference $\wtv$ with the maximum regret using an LP algorithm
	and adds $\kPoint_1(\wtv,\PntSet)$ to the regret set. We use Gurobi software \cite{gurobi}
  	to solve the LP problems efficiently. We remark that this algorithm, as a preprocessing,
	removes all data points that are not on the skyline.

\smallskip

  \textbf{\textit{\GreedyKII{}}} is the extension of the \GreedyKI{} algorithm for $k>1$, proposed by
  \cite{chester2014computing}, and it is the state of the art for the $k$-\Probl{}. In \cite{chester2014computing}
  they discard all the points not on the skyline as preprocessing to run the experiments.
  The CTVW algorithm solves many (in the worst-case, $\Omega(n)$) instances of large LP
  	programs to add the next point to the regret set. The number of LP programs is
	controlled by a parameter $T$---a larger $T$ increases the probability of adding
	a good point to the regret set, but also leads to a slower algorithm.
  In the original paper, the authors suggested a value of $T$ that is exponential in $k$;
  for instance, $T \geq 2.4\cdot 10^7$ for $k=10$, which is clearly not practical.
	In practice, Chester et al. \cite{chester2014computing} used $T=54$ for $k=4$,
	which is also the value we adopted in our experiments for comparison.
	Indeed, using $T>54$ increases the running time but does not lead to significantly better
	regret sets.

\smallskip

The algorithms are implemented in C++ and we run on a $64$-bit machine with four $3600$ MHz cores
and $16$GB of RAM with Ubuntu $14.04$.

In evaluating the quality $\ell_k(\SbSet)$ of a regret set $\SbSet\subseteq \PntSet$,
we compute the regret for a large set of random preferences (for example for $d=3$ we take $20000$
preferences), and use the maximum value found as our estimate. In fact, this approach
gives us the distribution of the regret over the entire set of preference vectors.

\mparagraph{Datasets}
We use the following datasets in our experiments, which include both synthetic and real-world.

 \textbf{\textit{BB}}
 is the basketball dataset \footnote{\texttt{databasebasketball.com}} that has been
 widely used for testing algorithms for skyline computation, top-$k$ queries,
 and the $k$-\Probl problem, \cite{chester2014computing, jasna2014algorithm, kulkarni2015skyline, kulkarni2016parallel, vlachou2010reverse}.
 Each point in this dataset represents a basketball player and its coordinates contain
five statistics (points, rebounds, blocks, assists, fouls) of the player.

\smallskip

\textbf{\textit{ElNino}}
	is the ElNino dataset \footnote{\texttt{archive.ics.uci.edu/ml/datasets/El+Nino}}
	containing oceanographic data such as wind speed, water temperature, surface temperature etc.
	measured by buoys stationed in the Pacific ocean, and also used in \cite{chester2014computing}.
	This dataset has some missing values, which we have filled in with the minimum value of
	a coordinate for the point. If some values are negative (where it does not make sense to
		have negative values) they are replaced by the absolute value.

\smallskip

\textbf{\textit{Colors}} is a data set containing the mean, standard deviation, and skewness of
	each $H$, $S$, and $V$ in the $HSV$ color space of a color
	image.\footnote{\texttt{www.ics.uci.edu/~mlearn/MLRepository.html}}
	This set is also a popular one for evaluating skylines and regret sets
	(see \cite{bartolini2008efficient, nanongkai2010regret}).

\smallskip

 \textbf{\textit{AntiCor}} is a synthetic set of points with \emph{anti-correlated} coordinates.
    Specifically, let $\hpr$ be the hyperplane with normal $\mathsf{n}=(1,\ldots, 1)$, and at distance
	$0.5$ from the origin.  To generate a point $\pnt$, we choose a random point $\widetilde{\pnt}$ on
		$\hpr\cap \posPoints$, a random number $t\backsim \mathcal{N}(0,\sigma^2)$, for a small standard deviation $\sigma$,
	and $\pnt=\widetilde{\pnt}+tn$.
	If $\pnt\in \posPoints$ we keep it, otherwise discard $\pnt$.
    	By design, many points lie on the skyline and the top-$k$ elements can differ significantly
    	even for nearby preferences.
	This data set is also widely used for testing top-$k$ queries or the skyline computation
	(see \cite{bks-skyline-01, nanongkai2010regret, vlachou2010reverse, morse2007efficient}).
    For our experiments we set $\sigma=0.1$ and generate $10000$ points.

\smallskip

 \textbf{\textit{Sphere}} is a set of points uniformly distributed on the unit sphere inside $\posPoints$,
	in which clearly all the points lie on the skyline. We generate the Sphere dataset with $15000$ points for $d=4$ (all points lie on the skyline).

\smallskip

 \textbf{\textit{SkyPoints}}
 	is a modification of the Sphere data set. We choose a small fraction of points from
	the Sphere data set and add, say, $20$ points that lie very close
	to $\pnt$ but are dominated by $\pnt$. For larger value of $k$, say $k>5$,
    considering only the skyline points is hard to decide which point is going to decrease
    the maximum regret ratio in the original set.
    We generate SkyPoints data set for $d=3$, $500$ points; with $100$ points on the skyline.

\begin{table}[t]
\centering
\begin{adjustbox}{width=0.45\textwidth}
\begin{tabular}{c | l | c | c | c}
ID & Description & d & n & Skyline \\ \hline
BB & Basketball & $5$ & $21961$ & $200$\\
ElNino & Oceanographic & $5$ & $178080$ & $1183$\\
Colors & Colors & $9$ & $68040$ & $674$\\
AntiCor & Anti-correlated points & $4$ & $10000$ & $657$\\
Sphere & Points on unit sphere & $4$ & $15000$ & $15000$ \\
SkyPoints & Many points close to skyline & $3$ & $500$ & $100$\\ \hline
\end{tabular}
\end{adjustbox}
\caption{Summary of datasets used in experiments.}
\tablelab{datasets}%
\end{table}

\smallskip

In evaluating the performance of algorithms, we focus on two main criteria, the runtime and
the regret ratio, but also consider a number of other factors that influence their performance
such as the value of $k$, the size of the skyline etc.

\CoreSet{} and \HS{} are both randomized algorithms so we report the average size
of the regret sets and the average running time computed over $5$ runs.
For $k=1$, we use the \GreedyKI{} algorithm, and for $k=10$, we use its extension,
the \GreedyKII{} algorithm.
In some plots there are missing values for the \GreedyKII{} algorithm,
because we stopped the execution after running it on a data set for $2$ days.

\begin{figure*}
\begin{subfigure}{0.25\textwidth}
  \includegraphics[width=\linewidth]{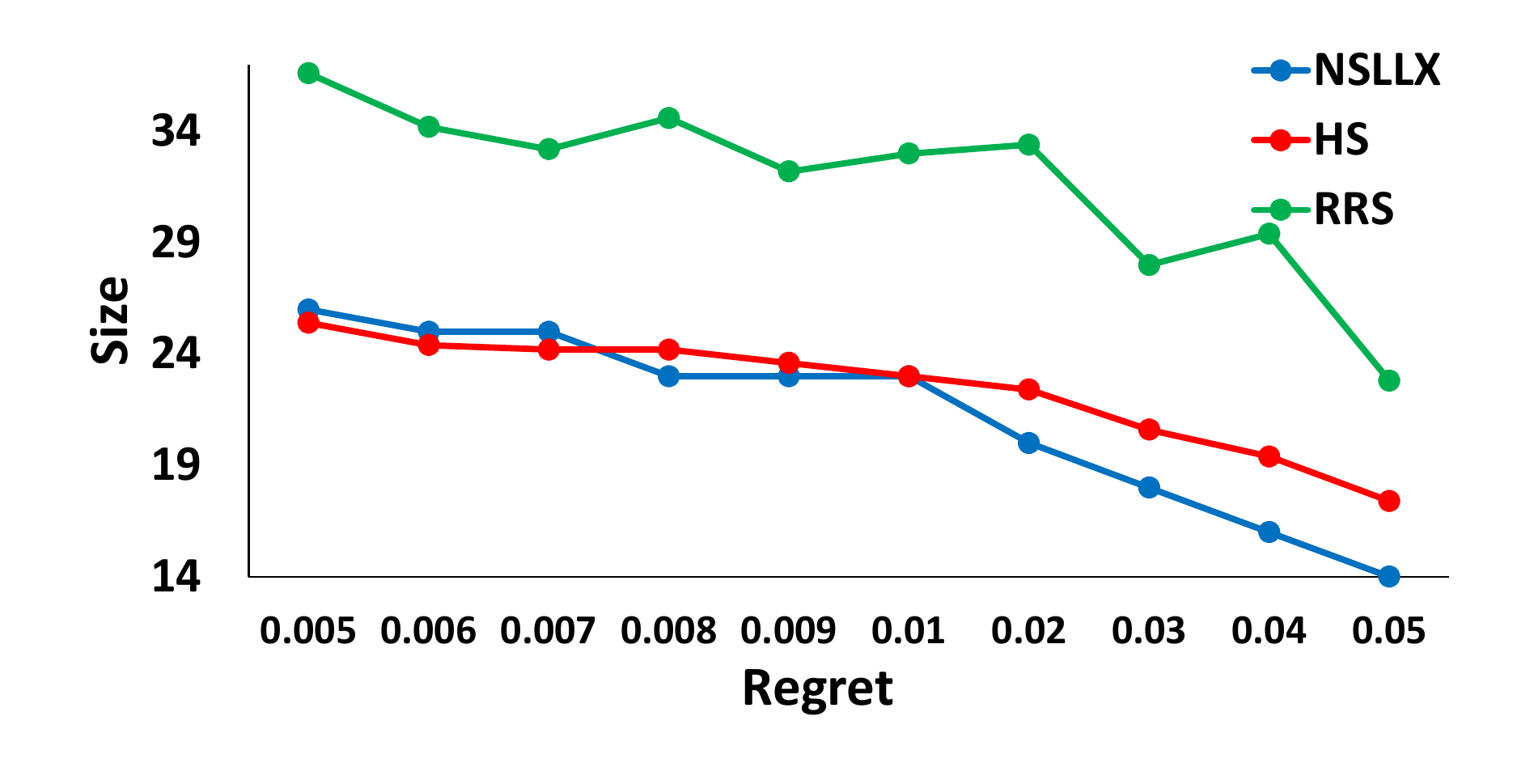}
  \caption{BB}\label{fig:BBM1}
\end{subfigure}\hfill
\begin{subfigure}{0.25\textwidth}
  \includegraphics[width=\linewidth]{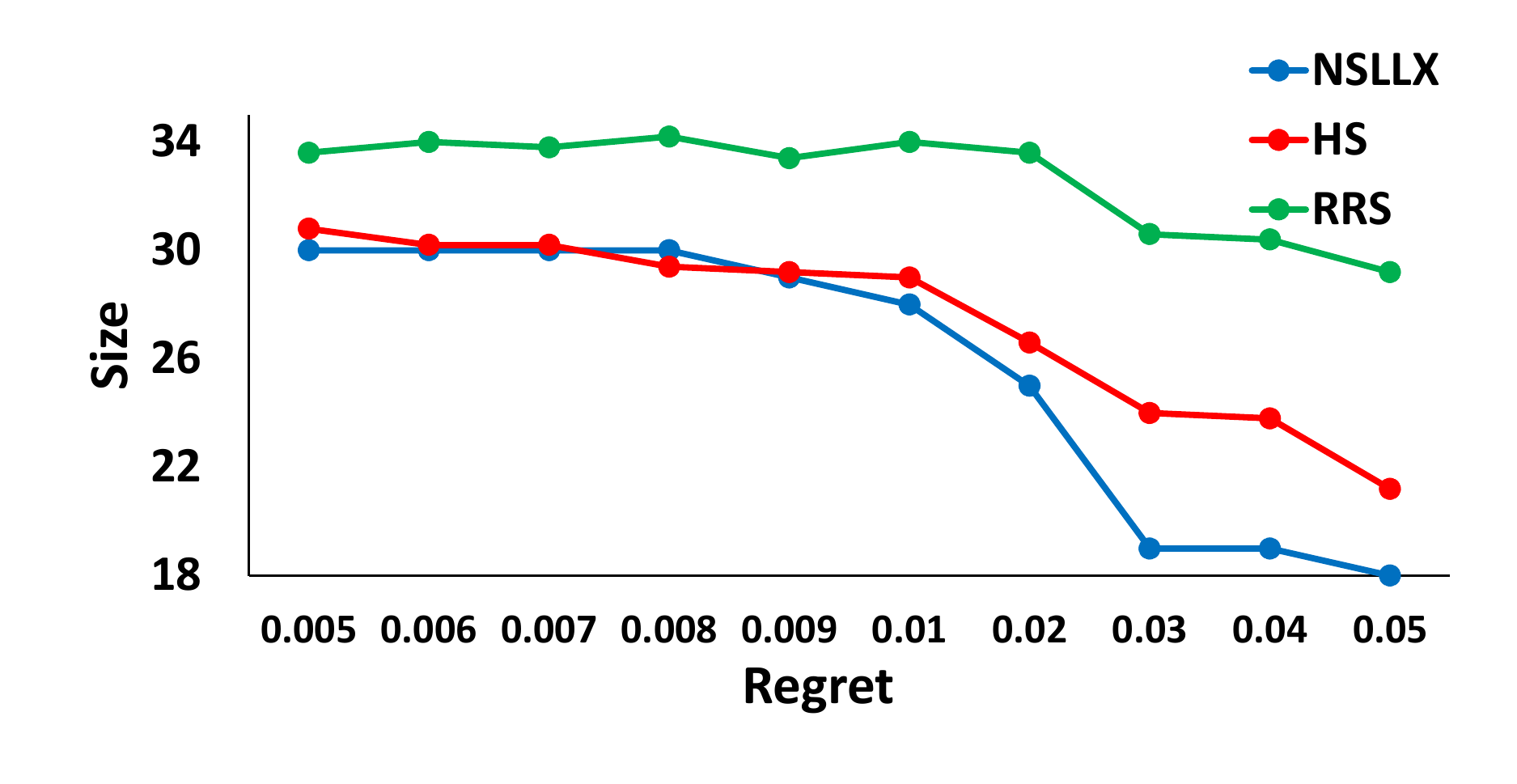}
  \caption{AntiCor}\label{fig:ACM1}
\end{subfigure}\hfill
\begin{subfigure}{0.25\textwidth}%
  \includegraphics[width=\linewidth]{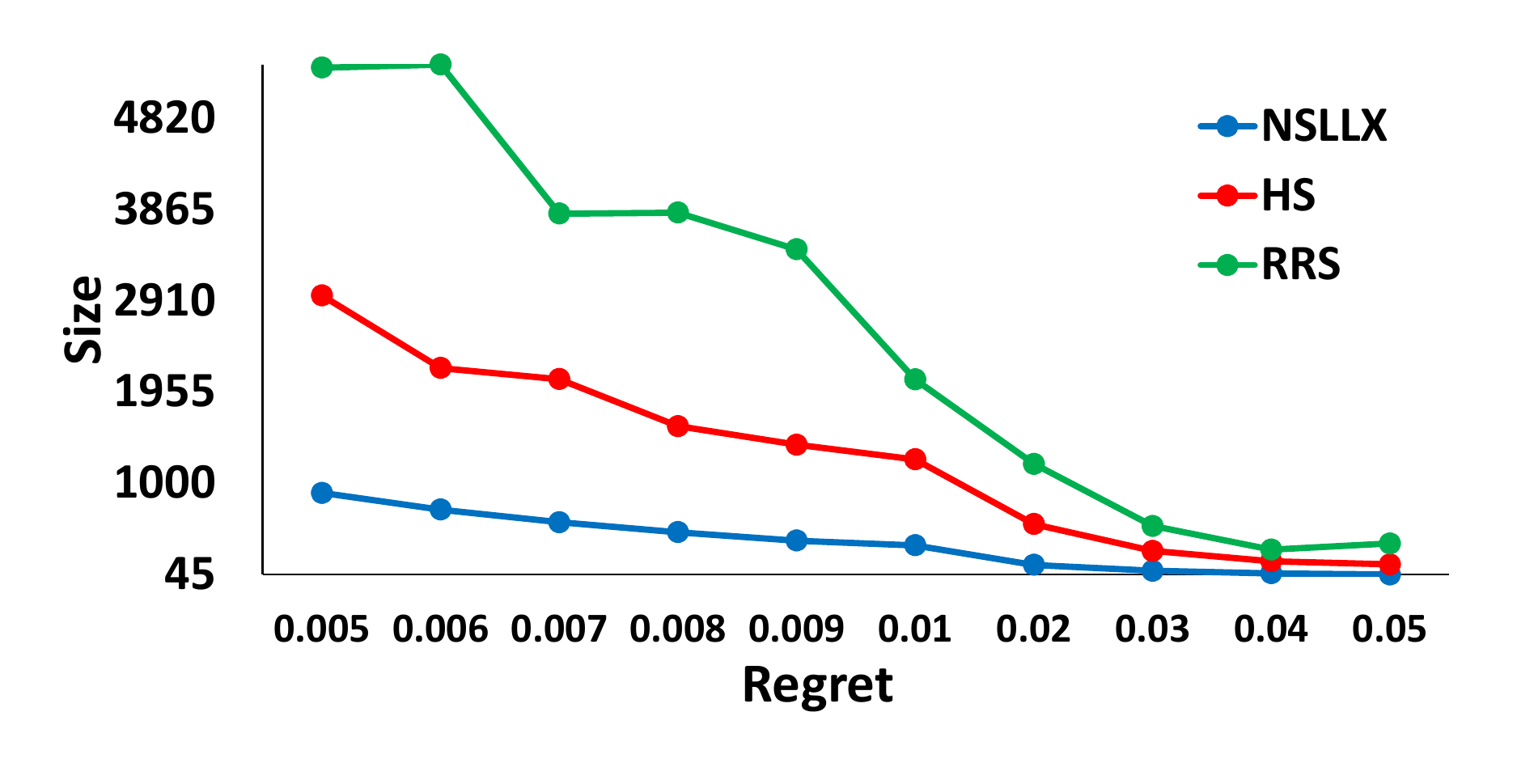}
  \caption{Sphere}\label{fig:SMK1}
\end{subfigure}\hfill
\begin{subfigure}{0.25\textwidth}%
  \includegraphics[width=\linewidth]{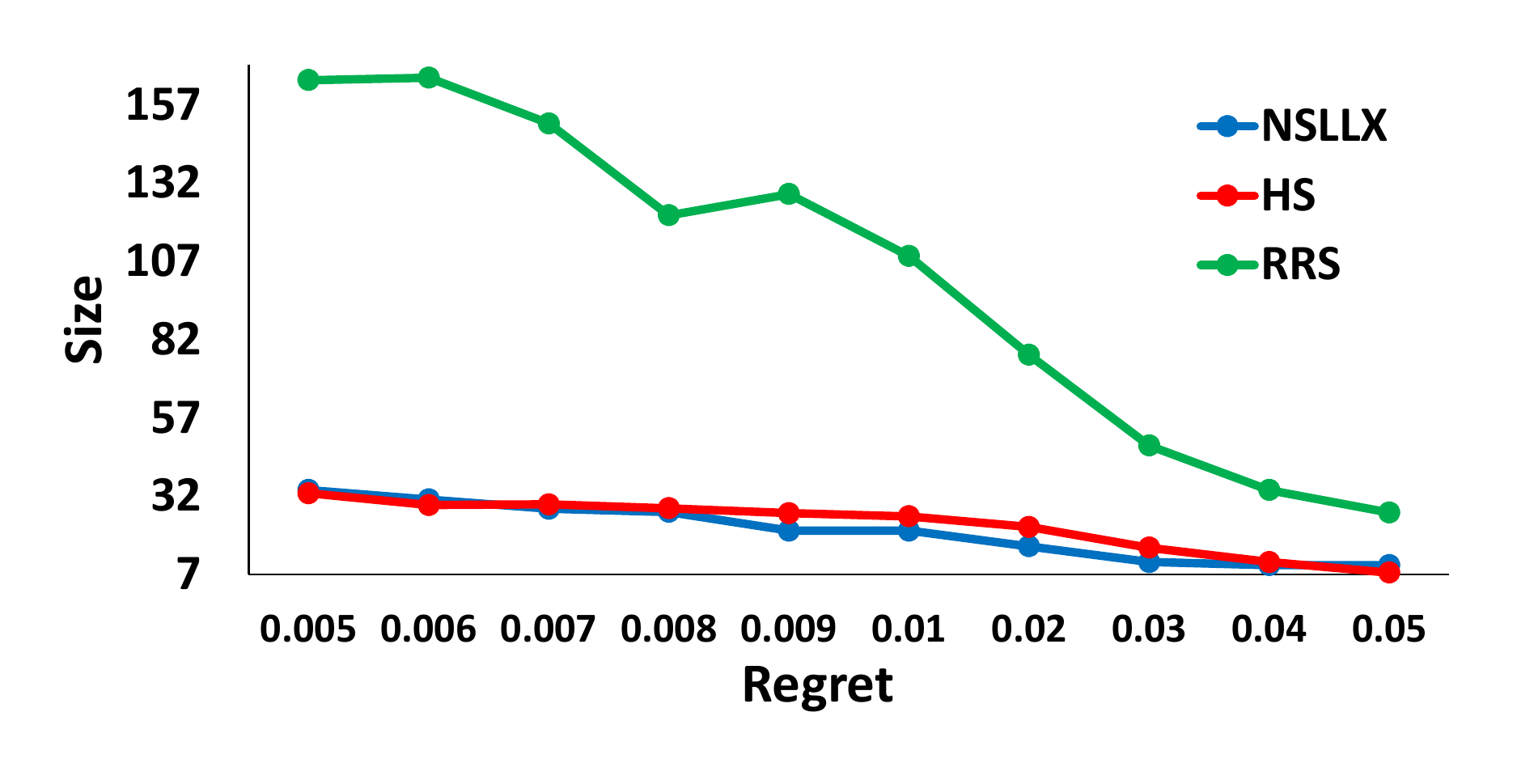}
  \caption{ElNino}\label{fig:ENMK1}
\end{subfigure}
\caption{Maximum regret ratio for $k=1$.}
\label{fig:MaxRegret1}
\end{figure*}

\begin{figure*}
\begin{subfigure}{0.25\textwidth}
  \includegraphics[width=\linewidth]{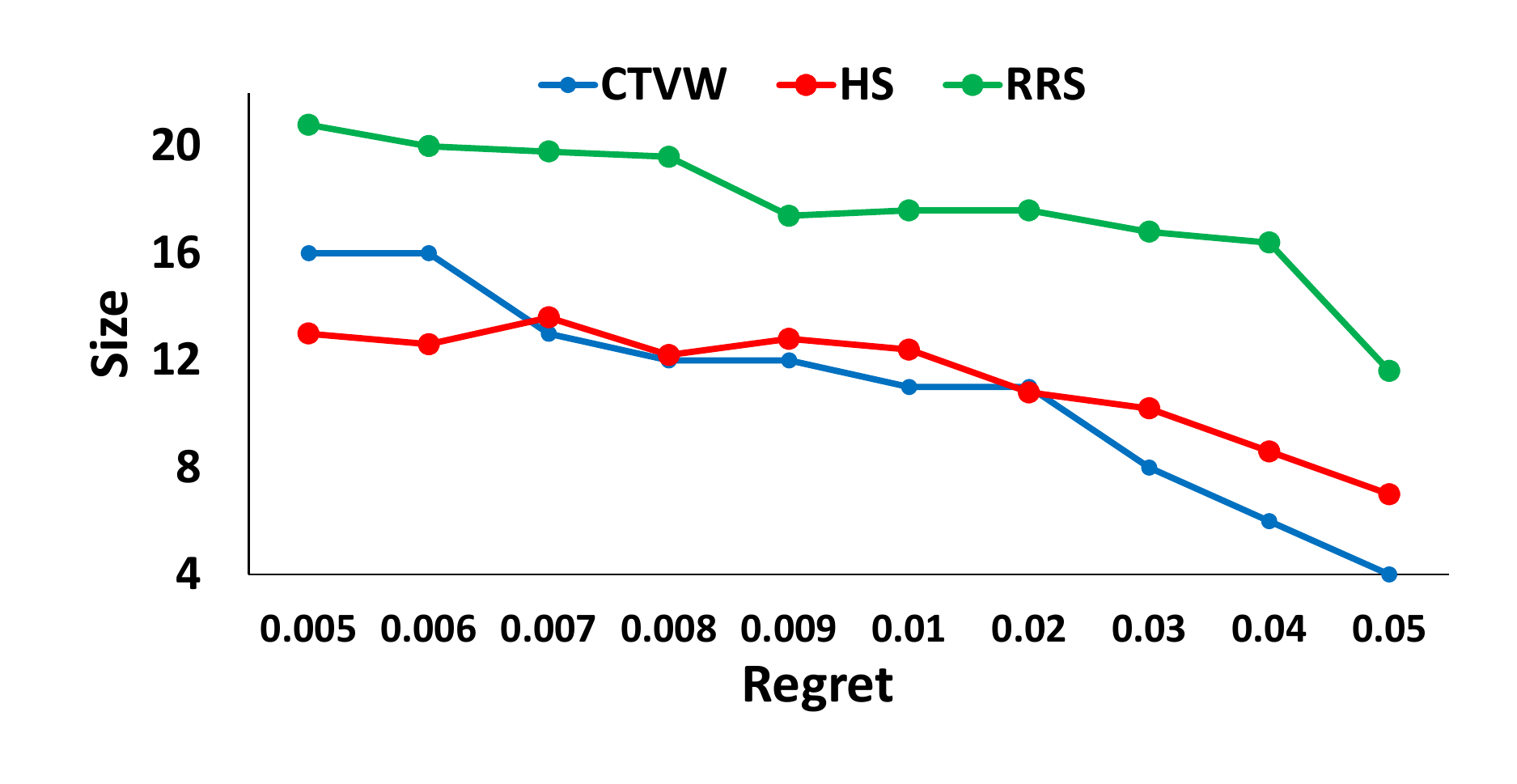}
  \caption{BB}\label{fig:BBM10}
\end{subfigure}\hfill
\begin{subfigure}{0.25\textwidth}
  \includegraphics[width=\linewidth]{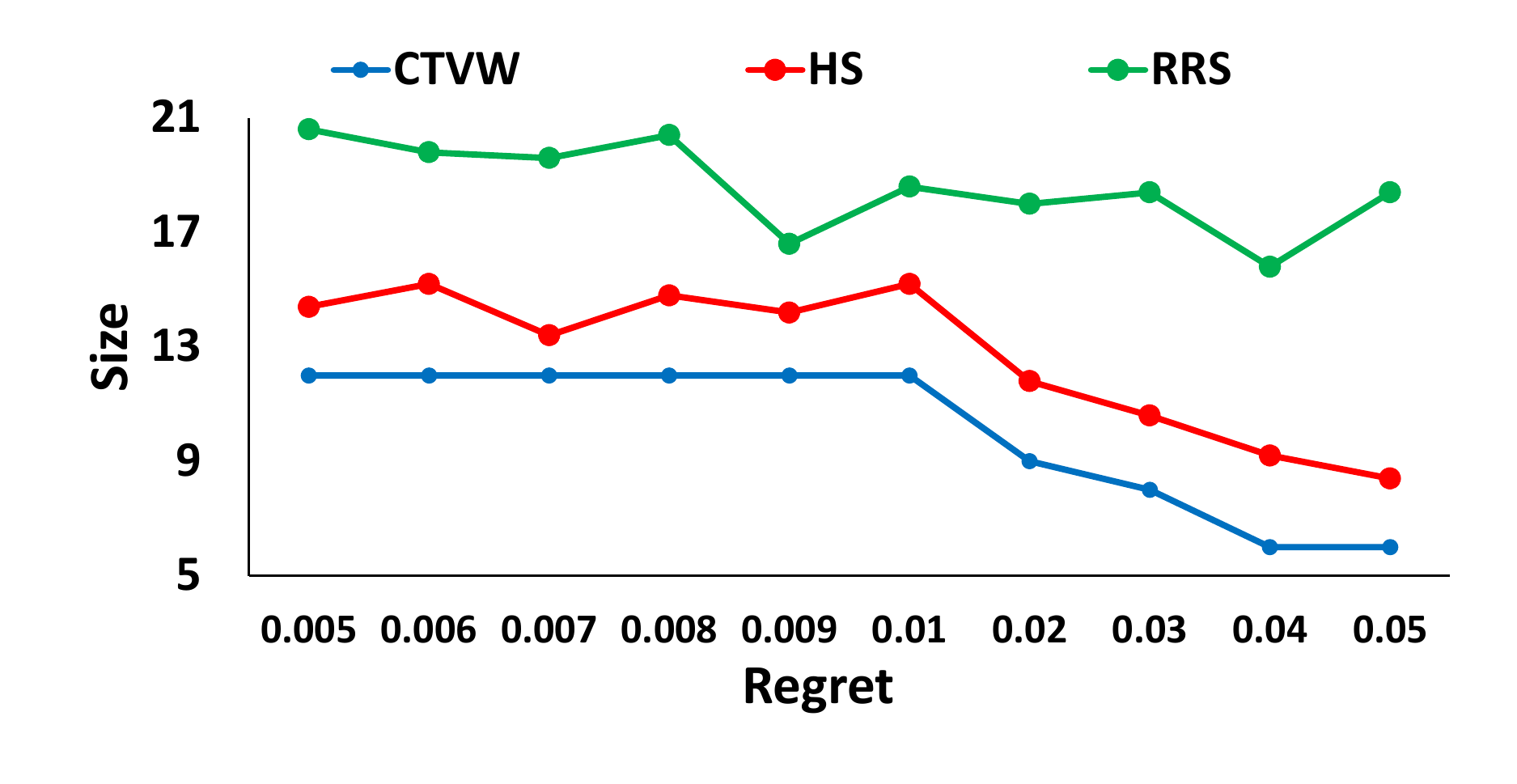}
  \caption{AntiCor}\label{fig:ACM10}
\end{subfigure}\hfill
\begin{subfigure}{0.25\textwidth}%
  \includegraphics[width=\linewidth]{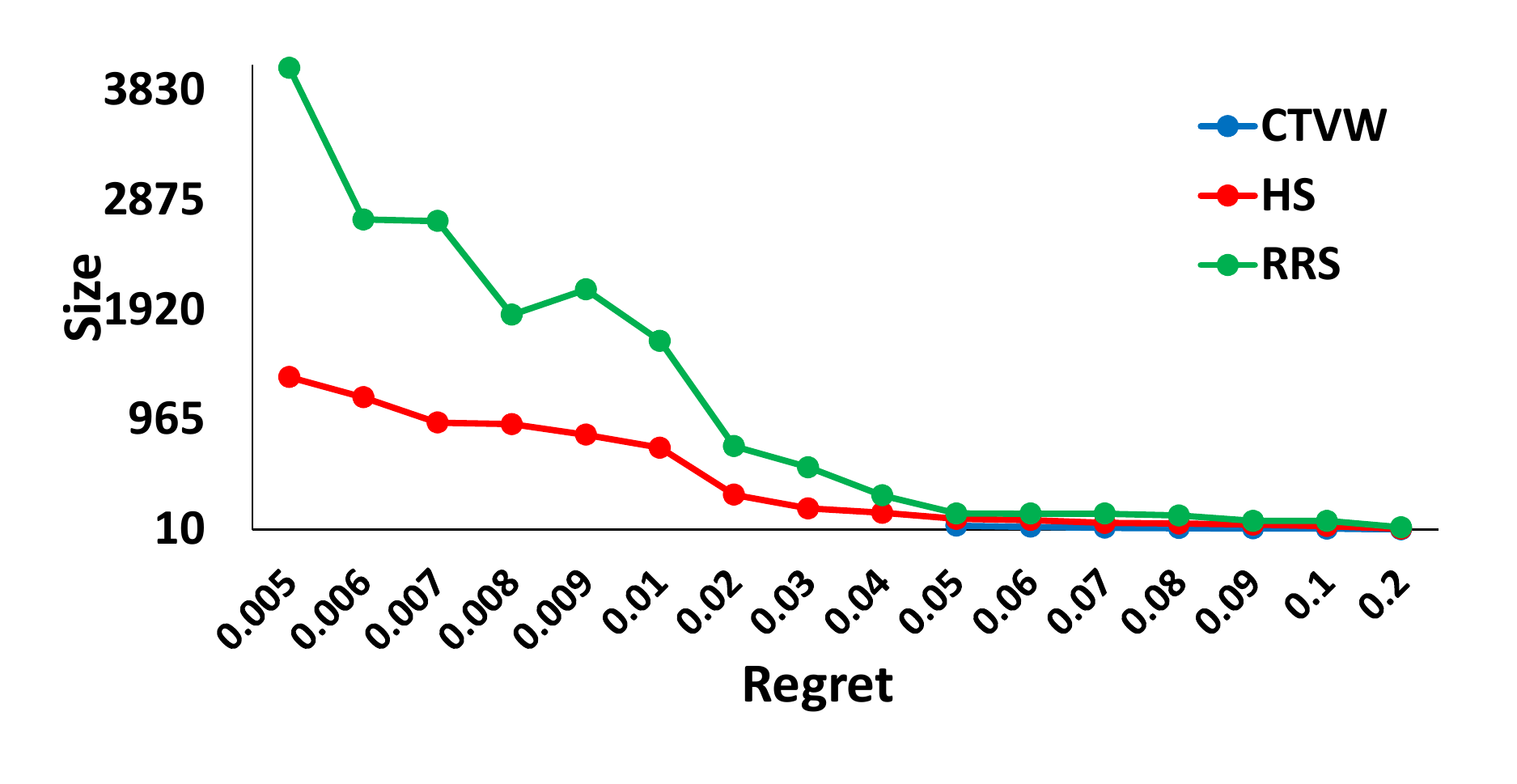}
  \caption{Sphere}\label{fig:SMK10}
\end{subfigure}\hfill
\begin{subfigure}{0.25\textwidth}%
  \includegraphics[width=\linewidth]{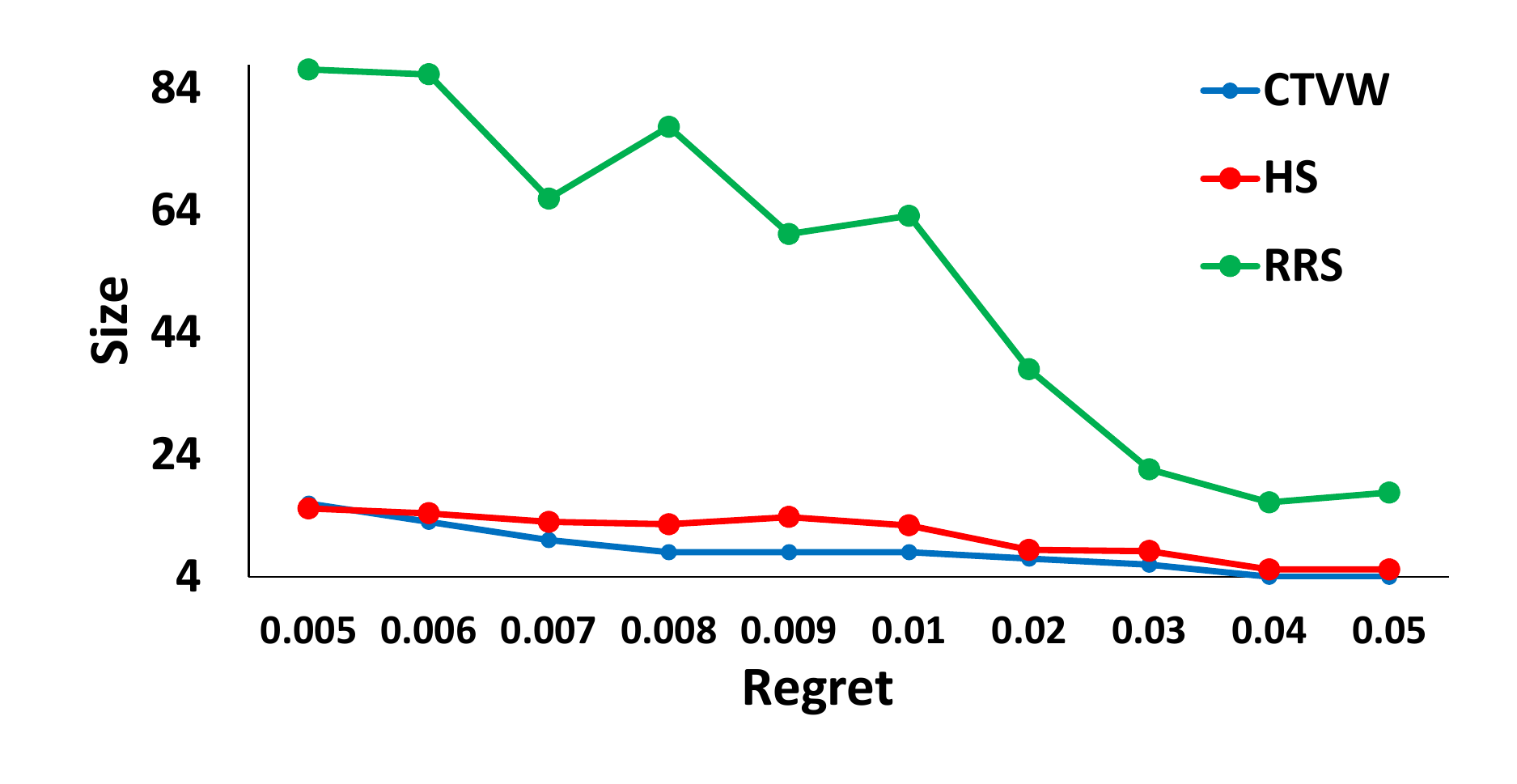}
  \caption{ElNino}\label{fig:ENMK10}
\end{subfigure}
\caption{Maximum regret ratio for $k=10$.}
\label{fig:MaxRegret10}
\end{figure*}

\mparagraph{Running time}
We begin with the runtime efficiency of the four algorithms, which is measured in the
number of seconds taken by each to find a regret set, given a target regret ratio.
\figref{fig:TimeK1} shows the running times of \GreedyKI{}, \HS{}, and \CoreSet{}
for $k=1$. The algorithm \CoreSet{} is the fastest.
For some instances, the running time of \HS{} and \CoreSet{} are close but in some other instances
\HS{} is up to three times slower.
The \GreedyKI{} algorithm is the slowest, especially for smaller values of the
regret ratio.
The relative advantage of our algorithms is quite significant for datasets that have
large skylines, such as AntiCor and Sphere.
Even for $k=1$, \GreedyKI{} is $7$ times slower than \HS{} on AntiCor data set and
$480$ times slower on Sphere data set, for regret ratio $\leq 0.01$.

The speed advantage of  \CoreSet{} and \HS{} algorithms over \GreedyKII{} becomes much more pronounced
for $k=10$, as shown in \figref{fig:TimeK10}.
Recall that \GreedyKII{} discards all points that are not on the skyline. The running time
is significantly larger if one runs this algorithm on the entire point set or when the skyline is
large. For example, for the AntiCor and
Sphere data sets, which have large size skylines, the \GreedyKII{} algorithm is
several orders of magnitude slower than ours.
If we set the parameter $\sigma=0.01$ for AntiCor data set, and generate $10000$ points(the skyline has $8070$ points in this case)
the running time of \GreedyKII{} is much higher as can be seen in Figure~\ref{fig:AntiCors001TimeK10}.

Because of the high running time of \GreedyKI{} and \GreedyKII{} algorithms, in Figure \ref{fig:LOGTimeSPhere},
we show the running time in the $\log$ scale with base $10$.

\begin{figure*}
\begin{subfigure}{0.48\textwidth}
  \includegraphics[width=\linewidth]{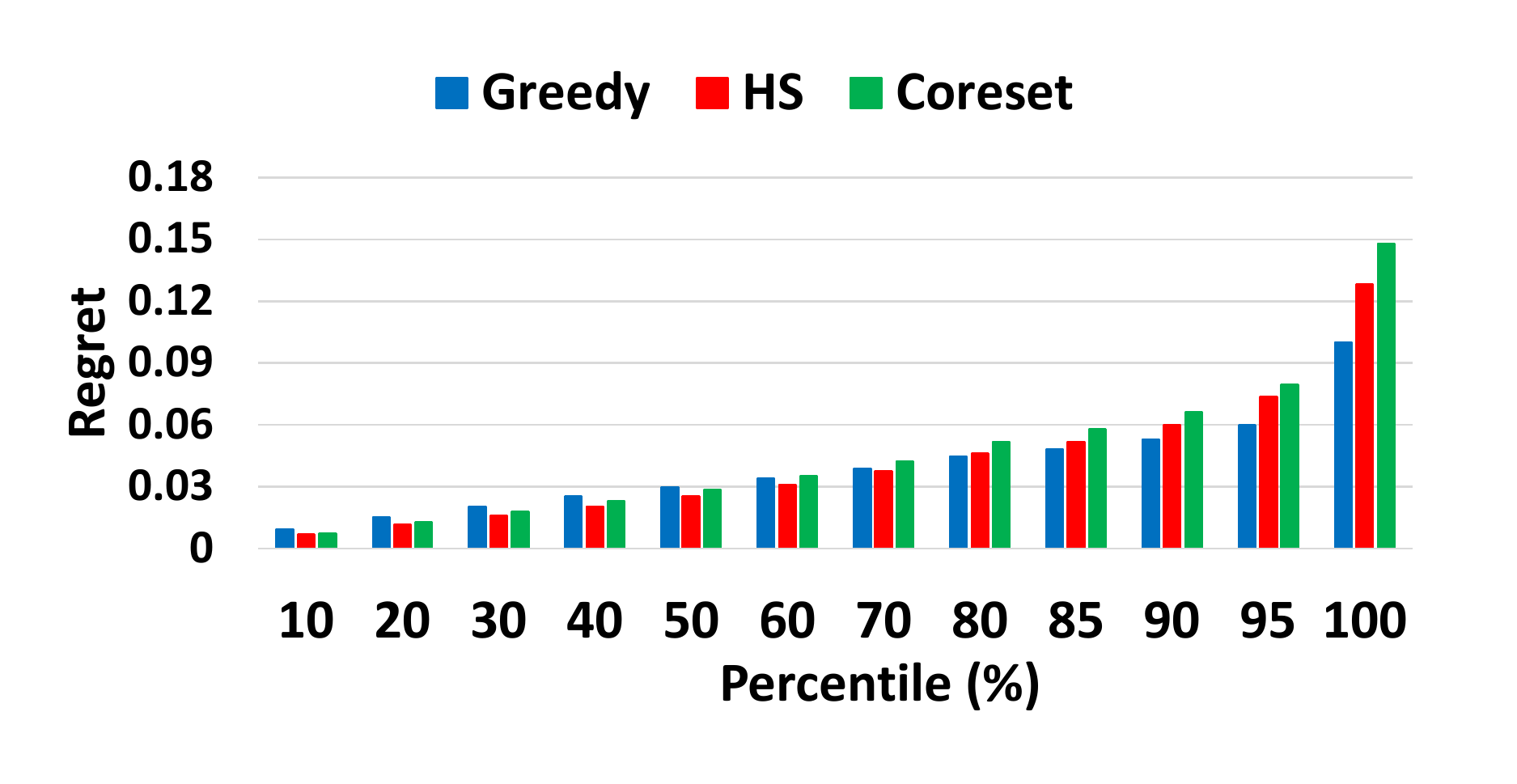}
  \caption{Sphere}\label{fig:SphereHist}
\end{subfigure}\hfill
\begin{subfigure}{0.48\textwidth}
  \includegraphics[width=\linewidth]{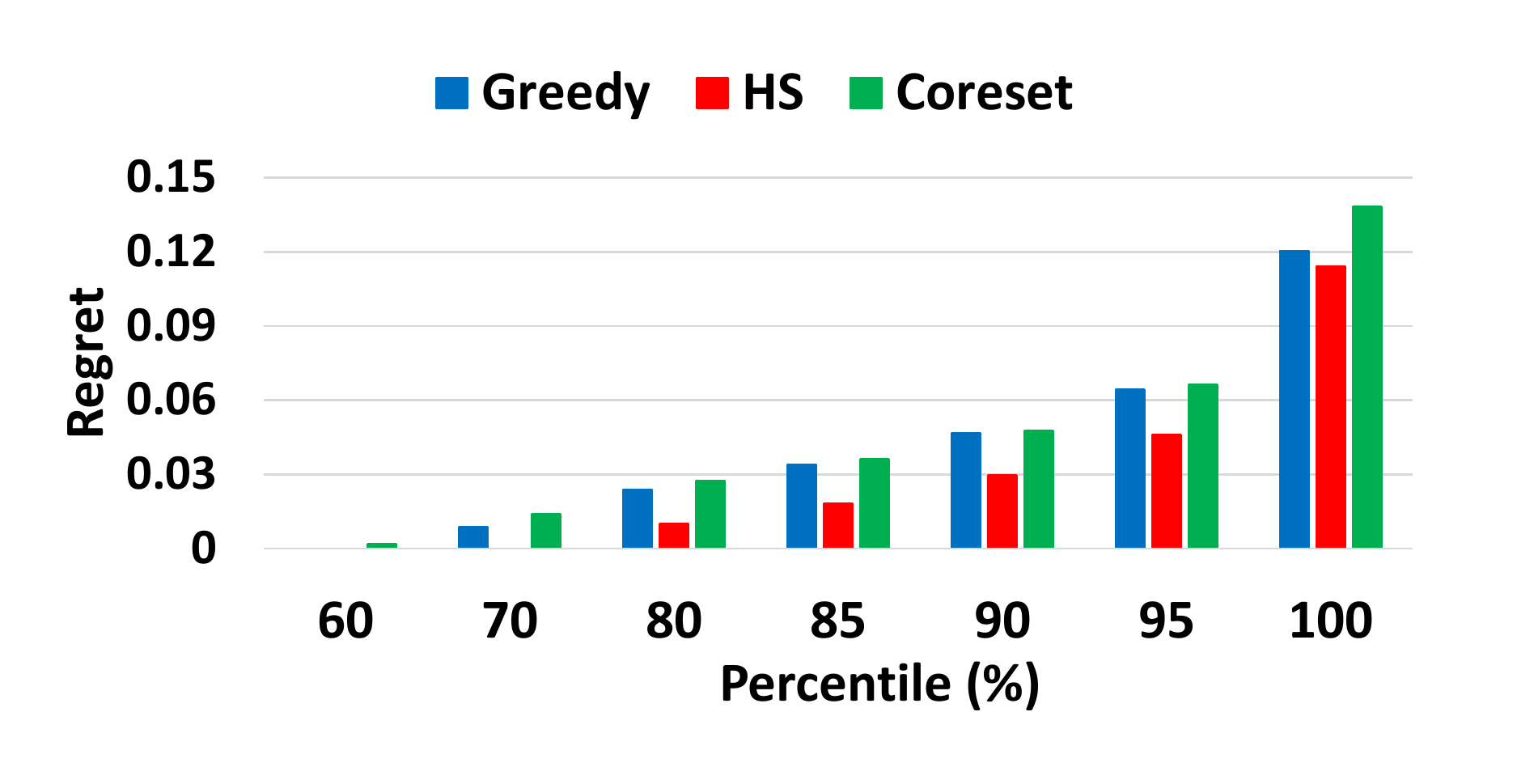}
  \caption{AntiCor}\label{fig:AntiCorHist}
\end{subfigure}
\caption{Regret distribution over the synthetic data sets, $k=1$.}
\figlab{fig:Hist}
\end{figure*}

\mparagraph{Regret ratio}
We now compare the quality of the regret sets (size) computed by the four algorithms.
Figures \ref{fig:MaxRegret1} and \ref{fig:MaxRegret10} show the results for
$k=1$ and for $k=10$, respectively.

The experiments show that in general the \HS{} algorithm finds regret sets comparable
in size to \GreedyKI{} and \GreedyKII{}.
This is also the case for AntiCor data set if we set $\sigma=0.01$ as can be seen in Figure~\ref{fig:AntiCors001K10}.
The \CoreSet{} algorithm tends to find the
largest regret set among the four algorithms, but it does have the advantage of
dynamic udpates: that is,  \CoreSet{} can  maintain a regret set under insertion/deletion
of points. However, since the other algorithms do not allow efficient updates, we do
not include experiments on dynamic updates.

The sphere data set is the worst-case example for regret sets since every point
has the highest score for some direction. As such, the size of the regret set
is much larger than for the other data sets. \HS{} and \CoreSet{} algorithm
rely on random sampling on preference vectors instead of choosing vectors adaptively
to minimize the maximum regret,
it is not surprising that for Sphere data sets \GreedyKII{} does $1.5$-$3$ times better than the \HS{} algorithm.
Nevertheless, as we will see below the regret of \HS{} in $95\%$ directions is close to that of \GreedyKII{}.

\mparagraph{Regret distribution}
The regret ratio only measures the \emph{largest} relative regret over all preference vectors.
A more informative measure could be to look at the entire distribution of the regret over
all preference vectors.
On all three real datasets, we found that $95\%$ of the directions had $0$-regret ratio for all
four algorithms. Therefore, we only show the results for the two synthetic main data sets, namely,
Sphere and AntiCor. See \figref{fig:Hist}.
In this experiment, we fixed the regret set size to $20$ for the Sphere dataset and $10$
for the AntiCor dataset. We observe that the differences in
the regret ratios in $95\%$ of the directions are much smaller than the differences in the maximum regret ratios.
For example, the difference of the maximum regret ratio between \CoreSet{} and \GreedyKI{} in
Sphere data set is $0.048$, while the difference in the $95\%$ ($85\%$) of the directions is $0.019$ ($0.0096$).
Similarly, for AntiCor data set the difference in the maximum regret ratio is $0.018$ but the difference
in the $95\%$ of the directions is $0.0021$.

\begin{figure}
\begin{subfigure}{0.24\textwidth}
  \includegraphics[width=\linewidth]{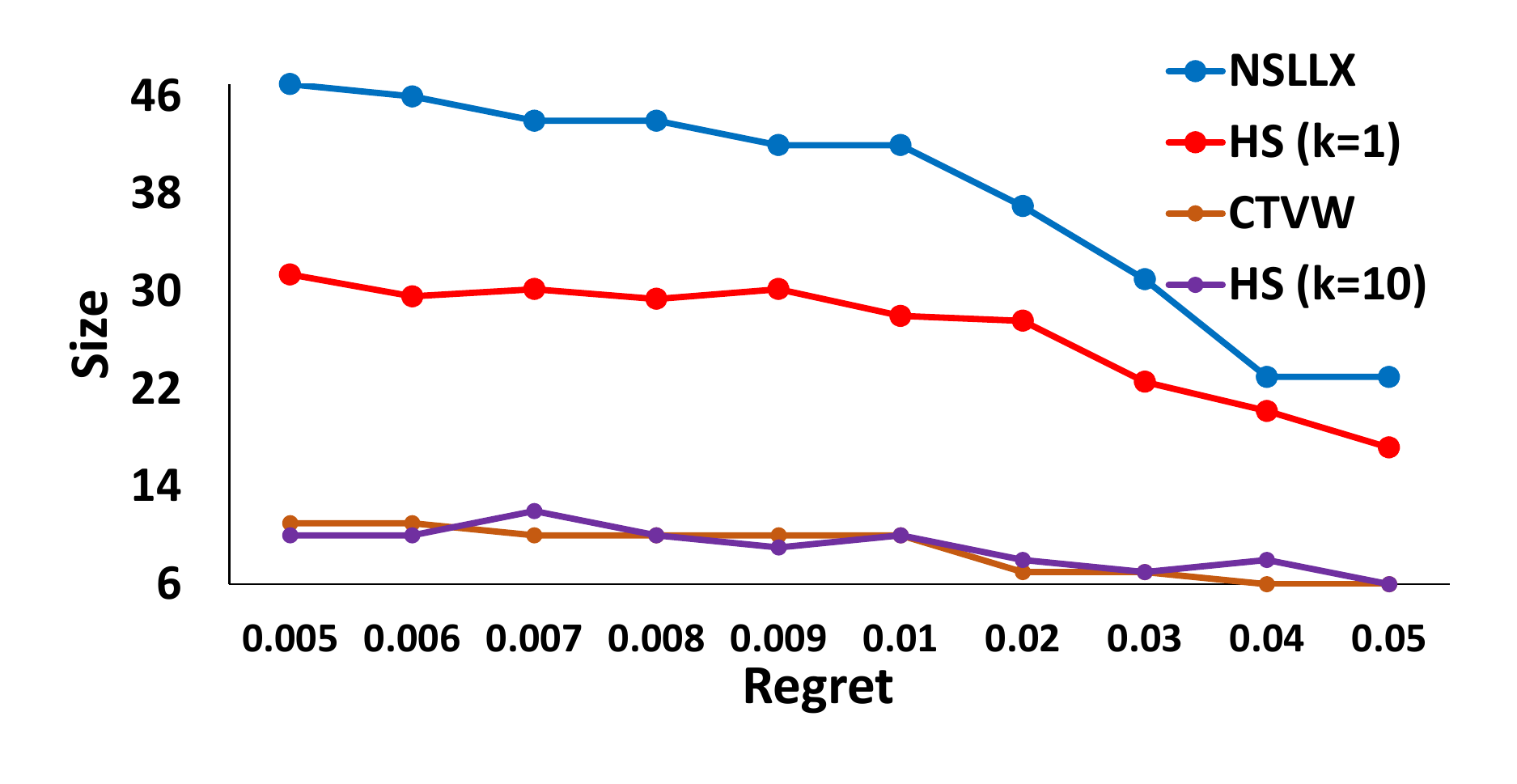}
  \caption{Regret ratio for $k=1, 10$.}\label{fig:ColorMax}
\end{subfigure}\hfill
\begin{subfigure}{0.24\textwidth}
  \includegraphics[width=\linewidth]{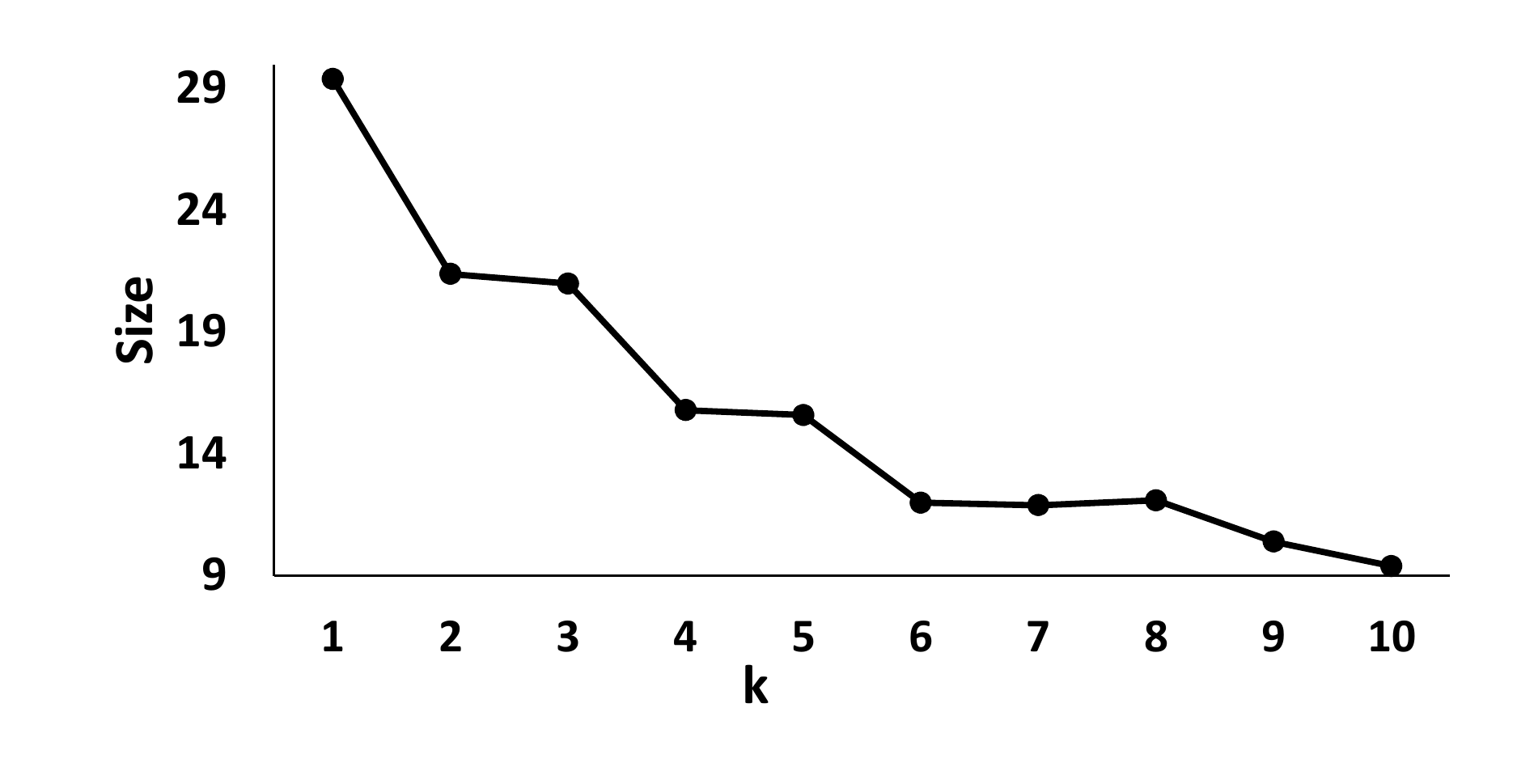}
  \caption{Size of the $(k,0.01)$-regret set as a function of $k$.}\label{fig:ColorK}
\end{subfigure}
\caption{Regret ratio and size of the regret set as a function of $k$ for Color data set.}\label{fig:ColorsP}
\end{figure}

\mparagraph{Impact of larger $k$}
We remarked in the introduction that the size of $(k,\epsilon)$-regret set can
be smaller for some datasets than their $(1,\epsilon)$-regret set, for $k>1$.
We ran experiments to confirm this phenomenon, and the results are shown in
Figure \ref{fig:ColorsP}.
As Figure \ref{fig:ColorMax} shows, the size of $1$-regret set is $3.5$ times larger
than $10$-regret sets for some values of the regret ratio. Figure \ref{fig:ColorK} shows how the size of the regret set
computed by the \HS{} algorithm decreases with $k$, for a fixed value of the regret ratio $0.01$.

\begin{figure}
\begin{subfigure}{0.24\textwidth}
  \includegraphics[width=\linewidth]{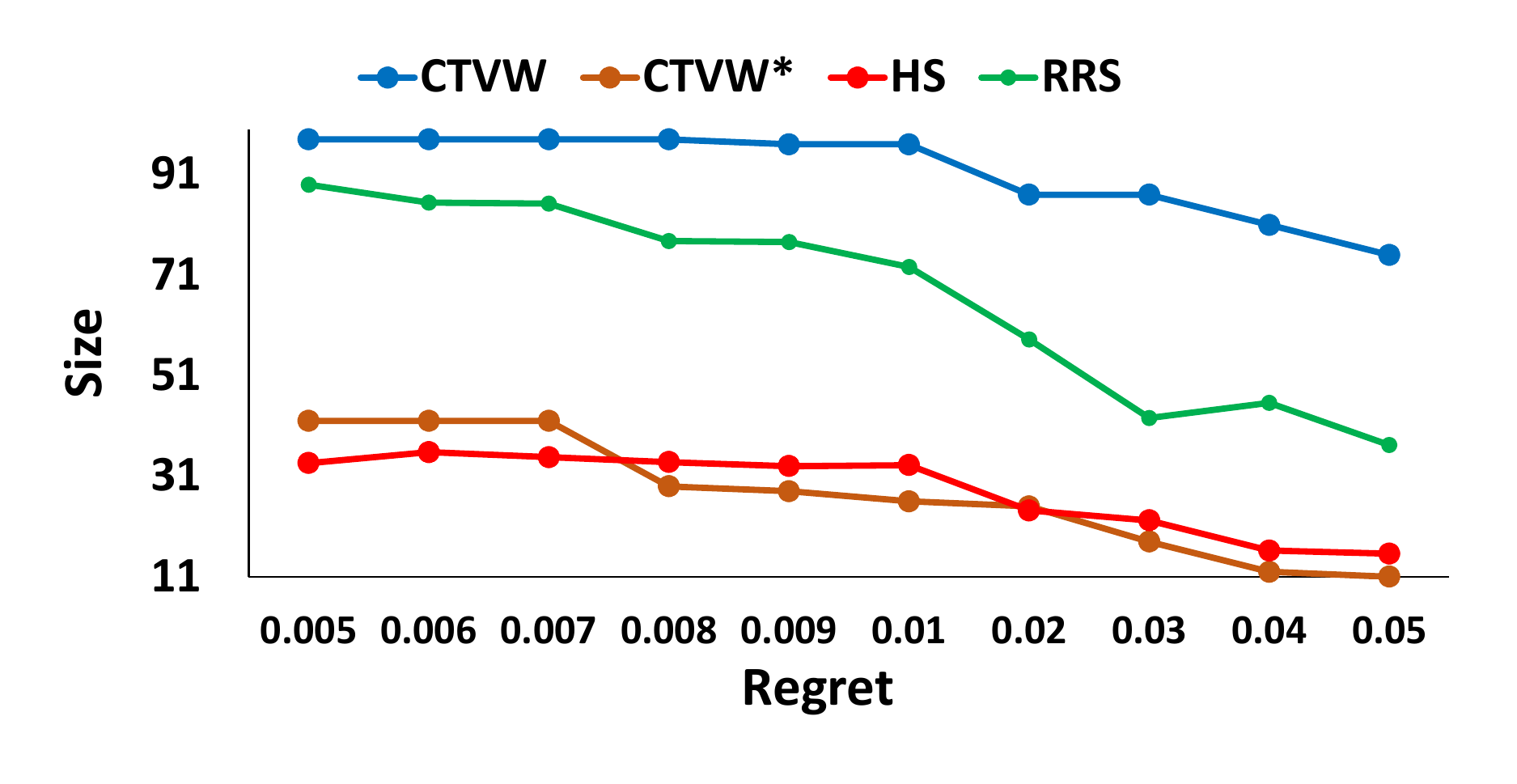}
  \caption{Regret ratio}\label{fig:SkyPointsM}
\end{subfigure}\hfill
\begin{subfigure}{0.24\textwidth}
  \includegraphics[width=\linewidth]{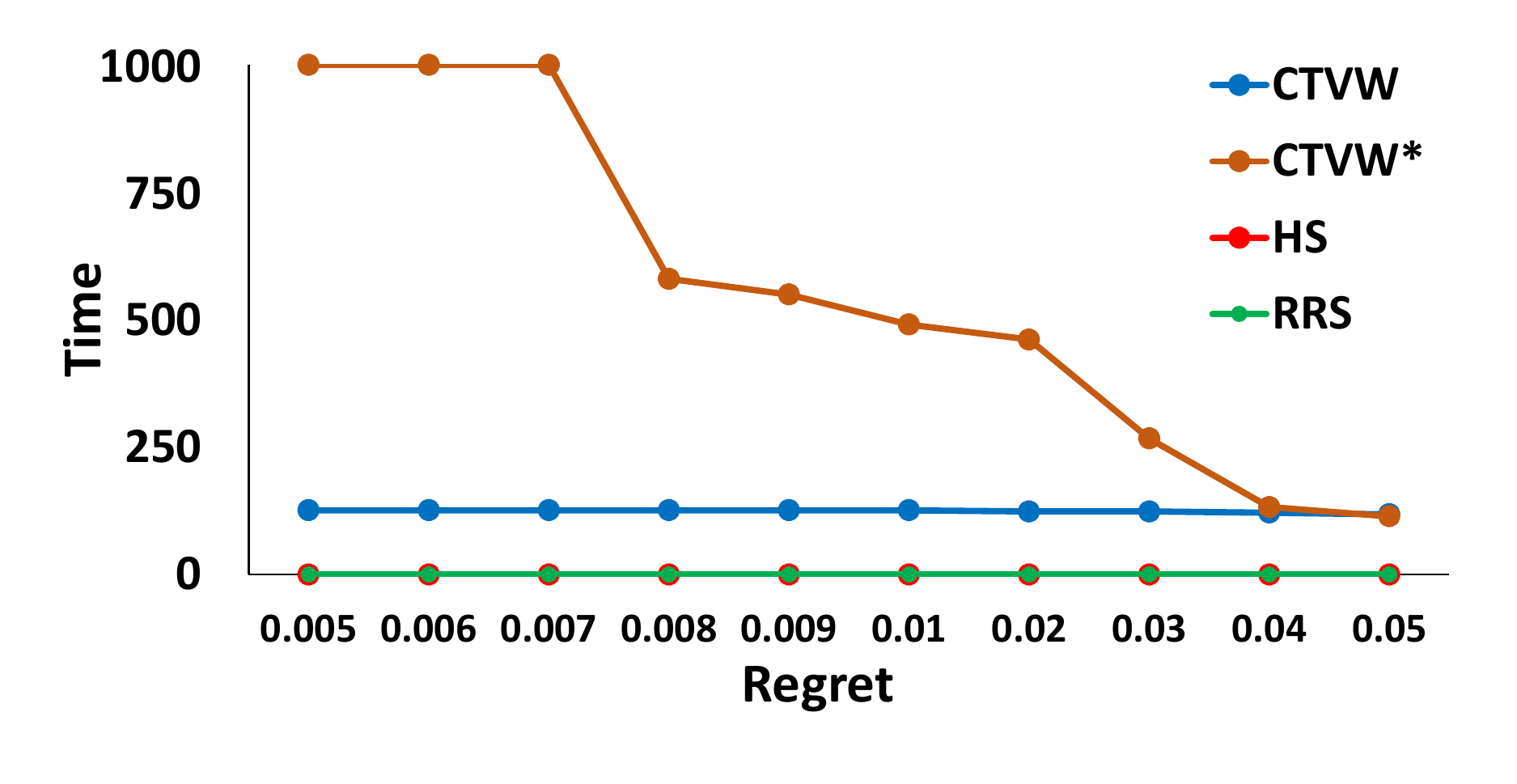}
  \caption{Running time}\label{fig:SkyPointsT}
\end{subfigure}
\caption{Regret ratio and running time of SkyPoints.}\label{fig:SkyPoints}
\end{figure}

\mparagraph{Skyline effect}
In order to improve its running time, the algorithm \GreedyKII{}~\cite{chester2014computing}
removes all the non-skyline points, as a preprocessing step, before computing the regret set.
While expedient, this strategy also risks losing good candidate points, and as a result may
lead to worse regret set. In this experiment, we used the Skypoint dataset to explore
this cost/benefit tradeoff. In particular, the modified version of \GreedyKII{} that
does not remove non-skyline points is called \GreedyKII{}*.

The results are shown in Figure \ref{fig:SkyPointsM}, which confirm that removal of non-skyline
points can cause significant increase in the size of the regret set, for a given target
regret ratio. The experiment shows that the regret set computed by \GreedyKII{}
is about $3$ times larger than the one computed by either \HS{} or \GreedyKII{}*.
(In this experiment, the regret size differences are most pronounced for small values
of regret ratio. When large values of regret ratio are acceptable, the loss of
good candidate points is no longer critical.)
Of course, while \GreedyKII{}* finds nearly as good a regret set as \HS{}, its
running time is much worse than that of \HS{}, or \GreedyKII{}, because of this change,
as shown in Figure \ref{fig:SkyPointsT}.

\section{Related Work}
\seclab{relatedWork}
The work on regret minimization was inspired by preference
top-$k$ and skyline queries.
Both of these research topics try to help a user
find the ``best objects'' from a database. Top-$k$ queries assign scores to
objects by some method, and
return the objects with the topmost $k$ scores while
the skyline query finds the objects such that no other object can be strictly better.
Efficiently answering top-$k$ queries has seen a long
line of work, see e.g. \cite{adhoc2007, guntzer2000, hybrid2010, prefer2001, li2006,marian2004,rahulefficient, theobald2004, xin2006, yu2012processing, zhang2006}
and the survey \cite{ilyas2008survey}.
In earlier work, the ranking of points was done by weight, i.e., ranking criterion was fixed. Recent work has
considered the specification of the ranking as part of the query. Typically, this is specified
as a preference vector $\wtv$ and the ranking of the points is by linear
projection on $\wtv$ see e.g. \cite{adhoc2007,prefer2001,yu2012processing}.
Another ranking criterion is based on the distance from a given query point in a metric
space i.e., the top-$k$ query is a $k$-nearest neighbor query \cite{tiakis2016}.

In general, preference top-$k$ queries are hard, and this has led to
approximate query answering \cite{chen2011efficient,yu2012processing,yu2016}.
Motivated by the need of answering preference top-$k$ queries,
Nanongkai et. al. \cite{nanongkai2010regret}
introduced the notion of a $1$-regret minimizing set (RMS) query.
Their definition attempted to combine preference top-$k$ queries
and the concept of skylines. They gave upper and lower bounds
on the regret ratio if the size of the returned set is fixed
to $r$. Moreover, they proposed an algorithm to compute a $1$-regret set of size $r$ with
regret ratio $O\left(\frac{d-1}{(r-d+1)^{1/(d-1)}+d-1}\right)$, as well
as a greedy heuristic that works well in practice.

Chester et. al. \cite{chester2014computing}
generalized the definition of $1$-RMS to the $k$-RMS
for any $k\geq 1$. They showed that the $k$-\Probl
is $\NP$-hard when the dimension $d$ is also an input to the problem,
and they provided an exact polynomial algorithm for $d=2$.
There has been more work on the $1$-RMS problem see
\cite{catallo2013top,nanongkai2012interactive,peng2014geometry},
including a generalization by Faulkner et. al. \cite{kessler2015k}
that considers non-linear utility functions.

The $1$-regret problem can be easily addressed by the notion of
$\epsilon$-kernel coresets, first introduced by
Agarwal et al. \cite{agarwal2004approximating}.
Later, faster algorithms were proposed to construct a coreset \cite{chan2004faster}.

The $1$-RMS problem is also closely related to the problem of approximating the
Pareto curve (or skyline) of a set of points.
Papadamitriou and Yannakakis \cite{papa-focs-00,papa-pods-01} considered this
problem and defined an approximate Pareto curve as a set of points whose
$(1+\epsilon)$ scaling dominates every point on the skyline.
They showed that there exists such a set of polynomial size \cite{papa-focs-00,papa-pods-01}.
However, computing such a set of the smallest size is $\NP$-Complete \cite{papa-tcs-07}.
See also \cite{yannakakis-tcs-05}.

\section{Conclusion}
\seclab{conclusion}

In this paper, we studied the \Probl. More specifically
we showed that the \Probl is $\NP$-Complete even in $\Re^3$,
which is a stronger result than the $\NP$-hardness proof in \cite{chester2014computing}
where the dimension is an input to the problem. Furthermore, we give bicriteria
approximation algorithms for the \Probl with theoretical guarantees, using the
idea of coresets and by mapping the problem to the well known hitting set problem.
Finally, we run experiments comparing the efficacy and the efficiency of our algorithms
with the greedy algorithms presented in \cite{chester2014computing, nanongkai2010regret}.

There are still some interesting problems for future work. In terms of
the complexity, our $\NP$-completeness proof holds for $k>1$. Is the $1$-regret minimization
problem $\NP$-Complete in $\Re^3$?
In terms of the approximation algorithms, is it possible to find algorithms with
theoretical guarantees where the running time does not have an exponential dependence on $d$, i.e.,
terms like $\frac{1}{\epsilon^{O(d)}}$ do not occur? This is important, because
in practice, the factor $\frac{1}{\epsilon^{O(d)}}$ can be very large even for moderately small
$d$ (say $d > 20$), thus severely limiting the practical utility of these algorithms.

\bibliographystyle{abbrv}
\bibliography{rms2}

\appendix
\seclab{appendix}

\section{Transform Polytope}
\label{Ap2}
We will present the transformation as a composition of transformations.
\mparagraph{Construction}
First, we translate $\Pi$ such that the origin $o$
is inside $\Pi$. Then, we compute the polar dual $\Pi^*$ (The polar dual of a polytope containing the origin $o$ is defined as the intersection of all hyperplanes $\dotp{x}{p} \leq 1$ where $p \in P$, and it can be equivalently defined as the intersection of the dual hyperplanes $\dotp{x}{v} \leq 1$ for all the vertices $v$ of $P$). Let $v$ be a vertex of $\Pi^*$.
Translate $\Pi^*$ such that $v$ becomes the origin. Then take a rotation such
that polytope $\Pi^*$ does not intersect the negative orthant --- i.e., the set of points in $\Re^3$ which have all coordinates strictly negative; we can always do it because
$\Pi^*$ is convex. Let $u_1$, $u_2$, $u_3$ be the three directions emanating from the origin
such that the cone defined by them, contains the entire polytope $\Pi^*$. Such directions always exist and can be found in polynomial time. It is known that
we can find in polynomial time an affine transformation such that $u_1$ is mapped to
the direction $e_1=(1,0.01,0.01)$, $u_2$ to direction $e_2=(0.01,1,0.01)$ and $u_3$ to $e_3=(0.01,0.01,1)$ (we can do it by first transforming $u_1, u_2, u_3$ to the unit axis vectors and then transform them to $e_1, e_2, e_3$).
Apply this affine transformation to $\Pi^*$ to get $\hat{\Pi}^*$.
Polytope $\hat{\Pi}^*$
lies in the first orthant, except for vertex $v$ which is at the origin. Shift this polytope slightly (such a shift can be easily computed in polynomial time) such that the origin lies in the interior of the polytope and very close to $v$ which now lies in the negative orthant, and all the other vertices are still in the first orthant.
Finally we compute the polar dual of $\hat{\Pi}^*$; call this $\hat{\Pi}$.
Translate $\hat{\Pi}$ until all vertices have positive coordinates, and let $\Pi'$ denote the new polytope.

\begin{lemma}
Polytope $\Pi'$ is combinatorially equivalent to $\Pi$ and satisfies properties (i), (ii).
\end{lemma}
\begin{proof}
We start by mapping property (ii) in the dual space.
Consider a polytope $G$ and its dual $G^*$ (where the origin lies inside them).
It is well known that any vertex $v$ of $G$ corresponds to a hyperplane $h_v$ in
the dual space that defines a facet of $G^*$.
An edge between two vertices in $G$ corresponds to an edge between the two corresponding faces in $G^*$.
Furthermore, if a vertex $v$ of $G$
is the top-$k$ vertex of $G$ in a direction $\wtv$, then the corresponding hyperplane $h_v$
is the $k$-th hyperplane (among the $n$ dual hyperplanes) that is intersected by the ray $o\wtv$, where $o$ is the origin.
From the above it is straightforward to map property (ii) in the dual space:
(ii') For any edge $(f_1, f_2)$ where $f_1, f_2$ are faces of $G^*$ there is
a direction $\wtv\in \posPoints$ such that the first two hyperplanes that are intersected by
the ray $o\wtv$ are $h_1$, $h_2$, where $h_1$ is the hyperplane that contains $f_1$ and $h_2$
the hyperplane that contains the face $f_2$.

We now show how these properties can be guaranteed in $\hat{\Pi}^*$.
Notice that from the construction of $\hat{\Pi}^*$, the origin lies inside $\hat{\Pi}^*$
and all faces of $\hat{\Pi}^*$ have non empty intersection with the positive octant.
By convexity, $\hat{\Pi}^*$ satisfies property (ii') because for any edge $e=(f_1,f_2)$ of $\hat{\Pi}^*$
there is a ray emanating from the origin that first intersects the edge $e$, and hence the hyperplanes $h_1, h_2$
are the first hyperplanes that are intersected by the ray.
So,
its dual polytope $\hat{\Pi}$ satisfies property (ii).
In addition, $\Pi^*$ is combinatorially equivalent to $\Pi$, by duality.
Since we apply an affine transformation
$\hat{\Pi}^*$ is also combinatorially equivalent to $\Pi^*$.
Finally, the polytope $\hat{\Pi}$ is combinatorially equivalent to $\hat{\Pi}^*$ (its dual).
Notice that translation does not change the combinatorial structure of a polytope or the ordering of the points in any direction, so
$\Pi'$ satisfies property (ii), property (i) by definition, and is also combinatorially equivalent to $\Pi$.
\end{proof}

\end{document}